\documentclass[11pt]{article}
\pdfoutput=1 
\usepackage[sc]{mathpazo}
\usepackage[utf8]{inputenc}
\usepackage[margin=1in]{geometry}
\usepackage{amsfonts}
\usepackage{amsmath}
\usepackage{amssymb}
\usepackage{amsthm}
\usepackage{float}

\usepackage{booktabs}
\usepackage[ruled]{algorithm2e}

\SetAlFnt{\small}
\SetAlCapFnt{\small}
\SetAlCapNameFnt{\small}
\SetAlCapHSkip{0pt}
\IncMargin{-\parindent}

\usepackage{hyperref}
\usepackage[svgnames]{xcolor}
\hypersetup{colorlinks={true},urlcolor={blue},linkcolor={DarkBlue},citecolor=[named]{DarkGreen}}
\usepackage[authoryear,square]{natbib}
\usepackage{doi}

\usepackage{microtype}
\usepackage[capitalise,nameinlink]{cleveref}
\usepackage{enumitem}

\crefname{enumi}{Condition}{Conditions}
\crefname{enumii}{Step}{Steps}

\usepackage{tikz}
\usetikzlibrary{arrows}
\usetikzlibrary{patterns,snakes}
\usetikzlibrary{decorations.shapes}
\tikzstyle{overbrace text style}=[font=\tiny, above, pos=.5, yshift=5pt]
\tikzstyle{overbrace style}=[decorate,decoration={brace,raise=5pt,amplitude=3pt}]
\usetikzlibrary{shapes.geometric}

\theoremstyle{definition}
\newtheorem{definition}{Definition}

\theoremstyle{plain}
\newtheorem{theorem}{Theorem}
\newtheorem{lemma}[theorem]{Lemma}
\newtheorem{corollary}[theorem]{Corollary}

\newtheorem{observation}{Observation}
\newtheorem{claim}{Claim}

\theoremstyle{definition}
\newtheorem{remark}{Remark}

\newcommand{\cd}{\textsc{Consensus-$1/k$-Division}}
\newcommand{\ch}{\textsc{Consensus-Halving}}

\newcommand{\plus}{I^{+}}
\newcommand{\minus}{I^{-}}

\newcommand{\tplus}{+_T}
\newcommand{\ttimes}{\times_T}

\newcommand{\linfixp}{\textsc{2D-Linear-FIXP}}
\newcommand{\tlinfixp}{\textsc{2D-Truncated-Linear-FIXP}}

\title{Consensus-Halving: Does It Ever Get Easier?\thanks{A preliminary version of this paper appeared in the proceedings of the 21st ACM Conference on Economics and Computation (EC 2020).}}

\author{
\begin{tabular}{cc}
& \\
\textbf{Aris Filos-Ratsikas} & \textbf{Alexandros Hollender}\\
\small{University of Edinburgh, United Kingdom} & \small{University of Oxford, United Kingdom}\\
\href{mailto:Aris.Filos-Ratsikas@ed.ac.uk}{\small{\texttt{Aris.Filos-Ratsikas@ed.ac.uk}}} & \href{mailto:alexandros.hollender@cs.ox.ac.uk}{\small{\texttt{alexandros.hollender@cs.ox.ac.uk}}}\\
& \\
\textbf{Katerina Sotiraki} & \textbf{Manolis Zampetakis}\\
\small{University of California, Berkeley, USA} & \small{University of California, Berkeley, USA}\\
\href{mailto:katesot@berkeley.edu}{\small{\texttt{katesot@berkeley.edu}}} & \href{mailto:mzampet@berkeley.edu}{\small{\texttt{mzampet@berkeley.edu}}}\\
& \\
\end{tabular}
}

\date{}

\begin{document}

\maketitle

\begin{abstract}
In the $\varepsilon$-\emph{Consensus-Halving} problem, a fundamental problem in fair division, there are $n$ agents with valuations over the interval $[0,1]$, and the goal is to divide the interval into pieces and assign a label ``$+$'' or ``$-$'' to each piece, such that every agent values the total amount of ``$+$'' and the total amount of ``$-$'' almost equally. The problem was recently proven by \citet{FRG18-Consensus,FRG18-Necklace} to be the first ``natural'' complete problem for the computational class PPA, answering a decade-old open question.

In this paper, we examine the extent to which the problem becomes easy to solve, if one restricts the class of valuation functions. To this end, we provide the following contributions. First, we obtain a strengthening of the PPA-hardness result of \citep{FRG18-Necklace}, to the case when agents have \emph{piecewise uniform} valuations with only \emph{two blocks}. We obtain this result via a new reduction, which is in fact conceptually much simpler than the corresponding one in \citep{FRG18-Necklace}. Then, we consider the case of \emph{single-block (uniform) valuations} and provide a parameterized polynomial time algorithm for solving $\varepsilon$-\emph{Consensus-Halving} for any $\varepsilon$, as well as a polynomial-time algorithm for $\varepsilon=1/2$. Finally, an important application of our new techniques is the first hardness result for a generalization of Consensus-Halving, the Consensus-$1/k$-Division problem \citep{SS03-Consensus}. In particular, we prove that $\varepsilon$-Consensus-$1/3$-Division is PPAD-hard.
\end{abstract}

\section{Introduction}

The topic of \emph{fair division} has been in the focus of research in economics and mathematics, since the late 1940s and the pioneering works of Banach, Knaster and Steinhaus \citep{Steinhaus1949}, who developed the associated theory. The related literature contains many interesting problems, with the most celebrated perhaps being the problems of \emph{envy-free cake-cutting} and \emph{equitable cake-cutting}, for which a plethora of results have been obtained. More recently, the computer science literature has made a significant contribution in studying the computational complexity of these problems, and attempting to design efficient algorithms for several of their variants \citep{aziz2016discrete,aziz2016discreteb,deng2012algorithmic,arunachaleswaran2019fair}. 

Another classical problem in fair division, whose study dates back to as early as the 1940s and the work of \citet{neyman1946theoreme}, is the \emph{Consensus-Halving problem} \citep{SS03-Consensus}. In this problem, there is a set of $n$ agents with valuation functions over the $I = [0,1]$ interval. The goal is to divide the interval into pieces using at most $n$ cuts, and assign a label from $\{+,-\}$ to each piece, such that every agent values the total amount of $I$ labeled ``$+$'' and the total amount of $I$ labeled ``$-$'' equally. The name ``Consensus-Halving'' is attributed to \citet{SS03-Consensus}, although the problem has been studied under different names in the past. For example, it is also known as \emph{The Hobby-Rice theorem} \citep{hobby1965moment}, or \emph{continuous necklace splitting} \citep{Alon87-Necklace}. Similarly to other well-known problems in fair division, the existence of a solution to the Consensus-Halving problem is \emph{always} guaranteed, and can be proven via the application of a fixed-point theorem; here the \emph{Borsuk-Ulam theorem} \citep{Borsuk1933}. 
As a matter of fact, the problem is a continuous analogue of the well-known Necklace Splitting problem \citep{Goldberg1985,Alon87-Necklace}, whose existence of a solution is typically established via an existence proof for the continuous version. 

The Consensus-Halving problem attracted attention in the literature of computer science recently, due to the recent results of \citet{FRG18-Consensus,FRG18-Necklace} who studied the computational complexity of the approximate version, in which there is a small allowable discrepancy $\varepsilon$ between the values of the two portions. First, in \citep{FRG18-Consensus}, the authors proved that $\varepsilon$-Consensus-Halving for inverse-exponential $\varepsilon$ is complete for the computational class PPA, defined by \citet{Papadimitriou94-TFNP-subclasses}. This was the first PPA-completeness result for a ``natural'' problem, i.e., a computational problem that does not have a polynomial-sized circuit explicitly in its definition, answering an open question from \citet{Papadimitriou94-TFNP-subclasses}, reiterated multiple times over the years \citep{Grigni2001,ABB15-2DTucker}. Then in \citep{FRG18-Necklace}, the authors strengthened their hardness result to the case of inverse-polynomial $\varepsilon$, which also established the PPA-completeness of the Necklace Splitting problem for $2$ thieves.

Despite the aforementioned results, the complexity of the problem is not yet well understood. Does the problem remain hard if one restricts attention to classes of simple valuation functions? Note that the reduction of \citep{FRG18-Consensus,FRG18-Necklace} uses instances with \emph{piecewise constant} valuation functions with \emph{polynomially many} pieces. On the opposite side, are there efficient algorithms for solving special cases of the problem? What if we allow a larger number of cuts?

\subsection{Our Results}
Towards understanding the complexity of Consensus-Halving, we present the following results.
\begin{itemize}
    \item[$\blacktriangleright$] We prove that $\varepsilon$-Consensus-Halving is PPA-complete, even when the agents have \emph{two-block uniform} valuations, i.e., valuation functions which are \emph{piecewise uniform} over the interval and assign non-zero value on at most \emph{two} pieces. This result holds even when $\varepsilon$ is inverse-polynomial, and extends to the case where the number of allowable cuts is $n + n^{1 - \delta}$, for some constant $\delta > 0$.\smallskip
    
    This is an important strengthening of the results in \citep{FRG18-Consensus,FRG18-Necklace} which require the agents to have piecewise constant valuations with polynomially many non-uniform blocks. En route to this result, we obtain a significant simplification to the proof of \citet{FRG18-Consensus,FRG18-Necklace}, which uses new gadgets for the encoding of the circuit of high-dimensional Tucker (see \cref{def:tucker}), which we reduce from. Our new reduction also gives a simplified proof of PPA-completeness for Necklace Splitting with $2$ thieves \citep{FRG18-Consensus,FRG18-Necklace}. \smallskip
    
    \item[$\blacktriangleright$] We study the case of \emph{single-block valuations} and provide the first algorithmic results for the problem.\footnote{To be precise, we provide the first such results for the version of the problem with $n$ agents and $n$ cuts. For a large number of cuts, \citet{brams1996fair} present algorithms for $\varepsilon$-approximate solutions. In fact, very recently and after the publication of the conference version of our paper, \citet{alon2020efficient} provided the first polynomial-time algorithms for the general $\varepsilon$-Consensus-Halving problem with a limited number of cuts (but still more than $2n$).}   Specifically, we present: \medskip
    \begin{itemize}
        \item[$\triangleright$] an algorithm for any $\varepsilon$, whose running time is polynomial in $1/\varepsilon$ and a parameter $d$ related to the maximum number of overlapping blocks.
        \item[$\triangleright$] a polynomial-time algorithm for $1/2$-Consensus-Halving. 
    \end{itemize}\medskip
    We complement our main results with a simple algorithm based on linear programming,
    which solves the problem for single-block valuations in polynomial-time, if one is allowed
    to use $2n-\ell$ cuts, for any constant $\ell$.\\
    
    \item[$\blacktriangleright$] As an application of the new ideas developed in our reduction, we obtain the first hardness result for a generalization of $\varepsilon$-Consensus-Halving, known as $\varepsilon$-\emph{Consensus-$1/k$-Division}, for $k \geq 3$. Specifically, we prove that $\varepsilon$-\emph{Consensus-$1/3$-Division} is PPAD-hard, when $\varepsilon$ is inverse-exponential.
\end{itemize}

\subsection{Discussion and Related Work}

The study of the Consensus-Halving problem dates back to the early 1940s and the work of \citet{neyman1946theoreme}. The first proof of existence for $n$ cuts can be traced back to the 1965 theorem of Hobby and Rice \citep{hobby1965moment}. The problem was famously studied in the context of Necklace Splitting, being a continuous analogue of the latter problem; in fact, most known proofs for Necklace Splitting go via the continuous version\footnote{This is true for the case of $2$ thieves. For $k$ thieves, the proofs go via the Consensus-$1/k$-Division problem instead.}  \citep{Goldberg1985,Alon1986}. The name \emph{Consensus-Halving} is attributed to \citet{SS03-Consensus}, who studied the continuous problem independently, and came up with a constructive proof of existence. Their construction, although yielding an exponential-time algorithm, was later adapted by \citet{filos2018hardness} to prove that the problem lies in the computational class PPA.

The class PPA was defined by \citet{Papadimitriou94-TFNP-subclasses} in his seminal paper in 1994, in which he also defined several other important subclasses of TFNP \citep{Megiddo1991}, the class of \emph{Total Search Problems in NP}, i.e., problems that \emph{always} have solutions which are efficiently verifiable. Among those classes, the class PPAD has been very successful in capturing the complexity of many interesting computational problems \citep{mehta2014constant,garg2018substitution,goldberg2019hairy,chen2013complexity}, highlighted by the celebrated results of \citet{Daskalakis2009} and \citet{chen2009settling} about the PPAD-completeness of computing a Nash equilibrium. On the contrary, since the definition of the class, PPA was not known to admit any natural complete problems, but rather mostly versions of PPAD-complete problems of a topological nature, defined on non-orientable spaces \citep{deng2016understanding,Grigni2001}. In 2015, \citet{ABB15-2DTucker} showed that the computational version of Tucker's Lemma \citep{tucker1945some}, already shown to be in PPA by \citet{Papadimitriou94-TFNP-subclasses}, is actually complete for the class.

Using the latter result as a starting point, \citet{FRG18-Consensus} proved that $\varepsilon$-Consensus-Halving is PPA-complete when $\varepsilon$ is inverse exponential. This was the first PPA-completeness result for a ``natural'' computational problem, where the term ``natural'' takes the specific meaning of a problem that does not have a polynomial-sized circuit in its definition. The quest for such problems that would be complete for PPA was initiated by Papadimitriou himself \citep{Papadimitriou94-TFNP-subclasses} and was later brought up again by several authors, including \citet{Grigni2001} and \citet{ABB15-2DTucker}. In the same paper, the authors also provided a computational equivalence between the $\varepsilon$-Consensus-Halving problem and the well-known Necklace Splitting problem of \citet{Alon87-Necklace} for $2$ thieves \citep{Goldberg1985,Alon1986}, when $\varepsilon$ is inverse-polynomial.
In \citep{FRG18-Necklace}, the authors strengthened their result to $\varepsilon$ being inverse-polynomial, which, together with the aforementioned result from \citep{FRG18-Consensus}, also provided a proof for the PPA-completeness of Necklace Splitting. As we mentioned earlier, besides being a strengthening, our PPA-hardness proof for $\varepsilon$-Consensus-Halving is a notable simplification over that of \citep{FRG18-Necklace}, and importantly, it holds for $\varepsilon$ which is inverse-polynomial. Therefore, we also obtain a new, simplified proof of PPA-hardness for Necklace Splitting with $2$ thieves.

For constant $\varepsilon$, the only hardness result that we know is the PPAD-hardness of \citet{filos2018hardness}, who also show that when $n-1$ cuts are allowed, deciding whether a solution exists is NP-hard. Recently, \citet{deligkas2019computing} studied the complexity of \emph{exact} Consensus-Halving and showed that the problem is FIXP-hard. Interestingly, the authors also introduced a new computational class, called BU (for Borsuk-Ulam) and showed that the problem lies in that class, leaving open the question of whether it is BU-complete. Very recently, using our new reduction presented in \cref{sec:ppahardness} as a starting point, \citet{BatziouHH21-consensus-BBU} showed that computing a \emph{strong} approximate solution of Consensus-Halving (with valuations represented by algebraic circuits) is complete for the class $\textup{BU}_a$ (the strong approximation version of the class BU).

If we generalize the number of labels to $\{1,2,\ldots,k\}$ rather that $\{+,-\}$, and we allow $(k-1)n$ cuts rather than only $n$, then we obtain a generalization of the Consensus-Halving problem which was referred to as \emph{Consensus-$1/k$-Division} in \citep{SS03-Consensus}. The existence of a solution for this problem can be proved via fixed-point theorems that generalize the Borsuk-Ulam theorem \citep{BSS81,Alon87-Necklace}, however very little is known about its complexity. One might feel inclined to believe that Consensus-$1/k$-Division is a harder problem that Consensus-Halving; however, note that in the former problem, we have more cuts at our disposal. In fact, \citet{FRG18-Necklace} conjectured that the complexities of the problems for different values of $k$ are incomparable, and are characterized by different complexity classes. The complexity classes that are believed to be the most related are called PPA-$k$, defined also by \citet{Papadimitriou94-TFNP-subclasses} in his original paper; we refer the reader to the recent papers of \citep{goos2019complexity,Hollender2019} for a more detailed discussion of these classes. As a matter of fact, in a recent paper we have shown \citep{FHSZ20} that the problem is in PPA-$k$, for any $k$ which is a prime power.

Before our paper, virtually nothing was known about the hardness of the problem when $k \geq 3$. While the techniques in \citep{FRG18-Necklace} were highly reliant on the presence of only two labels, our ideas do carry over to the case when $k=3$, which enables us to prove our PPAD-hardness result. While we do not expect the problem for $k\geq 3$ to be PPAD-complete, our proof offers important intuition about the intricacies of the problem and could be useful for proving stronger hardness results in the future.

\section{Preliminaries}\label{sec:preliminaries}

We start with the definition of the $\varepsilon$-approximate version of the \ch\ problem.

\begin{definition}[$\varepsilon$-\ch]
Let $k \geq 2$. We are given $\varepsilon > 0$ and a set $\mathcal{C}$ of continuous probability measures $\mu_1, \dots, \mu_n$ on $I=[0,1]$. The probability measures are given by their density functions on $I$. The goal is to partition the unit interval into $2$ (not necessarily connected) pieces $\plus$ and $\minus$ using at most $n$ cuts, such that $|\mu_j(\plus) - \mu_j(\minus)| \leq \varepsilon$ for all agents $j \in \{1,\ldots,n\}$.
\end{definition}

\noindent We will refer to the probability measures $\mu_1, \ldots, \mu_n$ as \emph{valuation functions} or simply \emph{valuations}. While the existence and PPA-membership results hold more generally, in this paper, we will restrict our attention to the case when the valuation functions are \emph{piecewise constant}. These can be represented explicitly in the input as endpoints and heights of value blocks.

\begin{definition}[\textsc{Piecewise constant valuation functions}]
  A valuation function $\mu_i$ is \emph{piecewise constant} over an interval $I$, if 
  the domain can be partitioned into a finite set of intervals such that the density of $\mu_i$
  is constant over each interval. 
\end{definition}

\noindent Piecewise constant functions are often referred to as \emph{step functions}. 

\begin{definition}[\textsc{Uniform valuation functions}]
We will consider the following subclasses of piecewise constant valuation functions.
\begin{itemize}
    \item[-] \textbf{\emph{Piecewise Uniform:}}, The domain can be partitioned into a finite set of intervals such that the density of $\mu_i$ is either $v_i$ or $0$ over each interval, for some constant $v_i$.\smallskip
    \item[-] \textbf{\emph{$d$-block Uniform:}} The domain can be partitioned into a finite set of intervals, such that in at most $d$ of those the density of $\mu_i$ is $v_i$ and everywhere else it is $0$, for some constant $v_i$.\smallskip
    \item[-] \textbf{\emph{$2$-block Uniform:}} $d$-block uniform valuations for $d=2$.\smallskip
    \item[-] \textbf{\emph{Single-block:}} $d$-block uniform valuations for $d=1$. Here we omit the term ``uniform'', as there is only a single value block.
\end{itemize}
\end{definition}
\noindent Obviously, piecewise constant $\supseteq$ piecewise uniform $\supseteq$ $2$-block uniform $\supseteq$ single-block. 

\subsection{The Computational Classes PPA and PPAD}\label{app:ppad}

As we mentioned in the introduction, \ch\ is a \emph{Total Search Problem in NP}, i.e., a problem with a guaranteed solution which is verifiable in polynomial time. The corresponding class is the class TFNP \citep{Megiddo1991}. Formally, a binary relation $P(x,y)$ is in the class TFNP if for every $x$, there exists a $y$ of size bounded by a polynomial in $|x|$ such that $P(x,y)$ holds and $P(x,y)$ can be verified in polynomial time. The problem is given $x$, to find such a $y$ in polynomial time.

The subclasses of TFNP that will be relevant for this paper are PPAD and PPA \citep{Papadimitriou94-TFNP-subclasses}. These are defined via their canonical problems, \textsc{End-of-Line} and \textsc{Leaf} respectively.

\begin{definition}[\textsc{End-of-Line}]
	The input to the \textsc{End-of-Line} problem consists of two Boolean circuits $S$ (for successor) and $P$ (for predecessor) with $n$ inputs and $n$ outputs such that $P(0^n)=0^n \neq S(0^n)$, and the goal is to find a vertex $x$ such that $P(S(x)) \neq x$ or $S(P(x))\neq x \neq 0^n$.
\end{definition}	
\noindent A problem is \emph{in} PPAD if it is polynomial-time reducible to {\sc End-Of-Line} and it is PPAD-\emph{complete} if {\sc End-Of-Line} reduces to it in polynomial-time. Intuitively, PPAD is defined with respect to a directed graph of exponential size, which is given \emph{implicitly} as input, via the use of the predecessor and successor circuits defined above. PPAD is a subclass of PPA, which is defined similarly, but with respect to an undirected graph and a circuit that outputs the neighbours of a vertex. Its canonical computational problem is called {\sc Leaf}, which is defined below.

\begin{definition}[\textsc{Leaf}]
The input to the \textsc{Leaf} problem is a Boolean circuit $C$ with $n$ inputs and at most $2n$ outputs, outputting the set $\mathcal{N}(y)$ of (at most two) neighbors of a vertex $y$, such that $|\mathcal{N}(0^n)|=1$, and the goal is to find a vertex $x$ such that $x \neq 0^n$ and $|\mathcal{N}(x)|=1$.
\end{definition}	
\noindent A problem is the class PPA if it is polynomial-time reducible to \textsc{Leaf} and is PPA-complete if \textsc{Leaf} reduces to it in polynomial time.

\subsection{High-dimensional Tucker}

Our reduction in \cref{sec:ppahardness} will start from the following problem, which is an $N$-dimensional variant of the $2D$-\textsc{Tucker} problem \citep{Papadimitriou94-TFNP-subclasses,ABB15-2DTucker}.

\begin{definition}[\textsc{high-D-Tucker}]\label{def:tucker}
An instance of \textsc{high-D-Tucker} consists of a labeling $\lambda: [8]^N \to \{\pm 1, \dots, \pm N\}$ computed by a Boolean circuit. We further assume that the labeling is antipodally anti-symmetric (i.e.\ for all $x$ on the boundary of $[8]^N$ it holds that $\lambda(\overline{x}) = - \lambda(x)$ where $\overline{x}_i = 9-x_i$ for all $i$), which can be enforced syntactically. A solution consists of two points $x,y \in [8]^N$ with $\lambda(x) = - \lambda(y)$ and $\|x-y\|_\infty \leq 1$.
\end{definition}

\noindent \citet{FRG18-Necklace} showed that the problem is PPA-hard, when the domain is $[7]^N$ instead of $[8]^N$. We adapt the hardness to the case of \cref{def:tucker} in the theorem below. 

\begin{theorem}\label{thm:tuckerppa}
\textsc{high-D-Tucker} is \textup{PPA}-complete.
\end{theorem}

\begin{proof}
\citet{Papadimitriou94-TFNP-subclasses} has shown that the problem lies in PPA. In order to show PPA-hardness, we use the fact that \citet{FRG18-Necklace} have proved that the problem is PPA-hard on the domain $[7]^N$ (instead of $[8]^N$) by using a standard snake-embedding technique \citep{chen2009settling,deng2017octahedral}.

Let $\lambda$ be an instance of \textsc{High-D-Tucker} but on the domain $[7]^N$ instead of $[8]^N$. We will reduce this to an instance $\lambda'$ of \textsc{High-D-Tucker} (on our standard domain $[8]^N$). In the two-dimensional case ($N=2$), the high level idea of the reduction is to take the domain $[7]^N$ and duplicate the central vertical and horizontal lines of the grid (thus also duplicating the labels at these grid points).

Formally, we proceed as follows. Define the operator $\widehat{\cdot}$ such that for any $r \in [8]$
\begin{equation*}\widehat{r} := \left\{ \begin{tabular}{lc}
    $r - 1$ &  \textup{if } $r \geq 5$\\
    $r$ & \textup{if } $r \leq 4$
\end{tabular} \right.
\end{equation*}
Then, for any $x=(x_1, \dots, x_N) \in [8]^N$, let $\widehat{x} = (\widehat{x}_1, \dots, \widehat{x}_N) \in [7]^N$. Now define $\lambda'$ such that for all $x \in [8]^N$, $\lambda'(x) := \lambda(\widehat{x})$. This is well-defined and given a circuit that computes $\lambda$, we can construct a circuit for $\lambda'$ in polynomial time.

Let us first show that if $\lambda$ is antipodally anti-symmetric on $[7]^N$, then $\lambda'$ is antipodally anti-symmetric on $[8]^N$. Consider any $x \in \partial([8]^N)$, i.e.\ there exists $j \in [N]$ such that $x_j \in \{1,8\}$. Note that we then have $\widehat{x} \in \partial([7]^N)$, because $\widehat{x}_j \in \{1,7\}$. Thus, we know that $\lambda(8-\widehat{x}_1, \dots, 8-\widehat{x}_N) = - \lambda(\widehat{x}_1, \dots, \widehat{x}_N)$. Using the key observation that $\widehat{9-x_i} = 8-\widehat{x}_i$ for all $i \in [N]$, we obtain that
$$\lambda'(9-x_1, \dots, 9-x_N) = \lambda(8-\widehat{x}_1, \dots, 8-\widehat{x}_N) = - \lambda(\widehat{x}_1, \dots, \widehat{x}_N) = -\lambda'(x_1, \dots, x_N).$$

It remains to show that given any solution to $\lambda'$, we can retrieve in polynomial time a solution to $\lambda$. Let $x,y \in [8]^N$ be such that $\lambda'(x) = -\lambda'(y)$ and $\|x-y\|_\infty \leq 1$. Then, we immediately obtain that $\lambda(\widehat{x}) = -\lambda(\widehat{y})$ and it remains to show that $\|\widehat{x}-\widehat{y}\|_\infty \leq 1$. Consider any $i \in [N]$. If $x_i,y_i \geq 5$ or if $x_i,y_i \leq 4$, then in both cases $|x_i-y_i| \leq 1$ implies $|\widehat{x}_i-\widehat{y}_i| \leq 1$. If $x_i \geq 5$ and $y_i \leq 4$, then $|x_i-y_i| \leq 1$ implies that $x_i=5$ and $y_i=4$ and we get $|\widehat{x}_i-\widehat{y}_i| = 0 \leq 1$. The remaining case is analogous. Thus, we have shown that $\widehat{x}, \widehat{y}$ form a solution to $\lambda$.
\end{proof}

\section{\ch\ with two-block uniform valuations is PPA-hard}\label{sec:ppahardness}

In this section, we present our first result, regarding the PPA-hardness of \ch. 

\begin{theorem}\label{thm:ppahardness}
$\varepsilon$-\ch\ is PPA-hard, when $\varepsilon$ is inverse-polynomial and the agents have two-block uniform valuations.
\end{theorem}

\noindent As we mentioned in the Introduction, \cref{thm:ppahardness} is a strengthening of the result of \citet{FRG18-Necklace}, which requires the valuation functions to have a polynomial number of value blocks, and which is seemingly very difficult to extend to two-block uniform valuations. To achieve this stronger result, we have to develop new gadgetry, based on a new interpretation of the cut positions with respect to the positions of points in the domain of \textsc{high-D-Tucker}. As it turns out, this new interpretation allows us to obtain a new proof of the main theorem of \citet{FRG18-Necklace}, one which is conceptually much simpler, even though it actually applies to more restricted valuations.

Before we proceed, we first remark the following. In \citep{filos2018hardness} (where the PPAD-hardness of $\varepsilon$-\ch\ was proven for constant $\varepsilon$) the authors presented a simple argument that allowed them to extend their hardness result to $n+c$ cuts, where $c$ is some constant. The idea is to make $c+1$ \emph{completely disjoint} copies of the instance of $\varepsilon$-\ch, and solve it using $n+c$ cuts. One of the copies would have to be solved using at most $n$ cuts, which is a PPAD-hard problem. We observe that the same principle applies generically (beyond PPAD-hardness and also to the results of \citet{FRG18-Consensus,FRG18-Necklace}), and in fact extends to $n+n^{1-\delta}$ cuts, where $\delta >0$ is some constant. From \cref{thm:ppahardness}, we obtain the following corollary.   

\begin{corollary}\label{cor:morecuts}
$\varepsilon$-\ch\ is PPA-hard, when $\varepsilon$ is inverse polynomial and the agents have two-block uniform valuations, even when one is allowed to use $n+n^{1-\delta}$ cuts, for constant $\delta >0$.
\end{corollary}

\noindent We are now ready to prove \cref{thm:ppahardness}. We first provide an overview of the reduction and we highlight the main simplifications over the proof of \citet{FRG18-Necklace}. Then we proceed to formally present the proof of \cref{thm:ppahardness}.

\subsection{Overview of the reduction}\label{sec:overview}

We are given an instance of \textsc{high-D-Tucker}, namely a labeling $\lambda: [8]^N \to \{\pm 1, \dots, \pm N\}$ computed by a Boolean circuit. We will show how to construct an instance of \textsc{Consensus-Halving} in polynomial time such that any $\varepsilon$-approximate solution yields a solution to the \textsc{high-D-Tucker} instance (for some inversely-polynomial $\varepsilon$). The complexity will be measured with respect to the representation size of the \textsc{high-D-Tucker} instance, i.e.,\ the size of the circuit $\lambda$ (which is also at least $N$).

For clarity and convenience, the instance of \textsc{Consensus-Halving} we will construct will not be defined on the domain $[0,1]$, but instead on some interval $[0,M]$, where $M$ is bounded by a polynomial in the size of the \textsc{high-D-Tucker} circuit $\lambda$. It is easy to transform this into an instance on $[0,1]$ by just re-scaling the valuation functions, namely scaling down the positions of the blocks by $M$ and scaling up the heights of the blocks by $M$.

\paragraph{\textbf{Overview.}} Let us first provide a very high-level description of the instance we construct. Similarly to \citep{FRG18-Necklace}, the left-most end of the instance will be the \emph{Coordinate-Encoding} region. In any solution $S$ to the instance, the way in which this region is divided amongst the labels $+$ and $-$ will represent a point $x \in [-1,1]^N$. A \emph{circuit-simulator} $C$ will read-in the coordinates of $x$, perform some computations (including a simulation of $\lambda$) and output $N$ values $[C(x)]_1, \dots, [C(x)]_N \in [-1,1]$. This circuit-simulator will be implemented by a set of agents where each agent will enforce one gate/operation of the circuit. Unfortunately, the circuit-simulator can sometimes fail to perform the desired computation, so instead of one circuit-simulator $C$ we will actually have a polynomial number $p(N)$ of circuit-simulators $C_1, \dots, C_{p(N)}$. Each of these circuit-simulators will be performing (almost) the same computation. Finally, we will introduce a \emph{Feedback region} where $N$ feedback agents $f_1, \dots, f_N$ will implement the \emph{feedback mechanism}. For each $i \in \{1, \dots, N\}$, feedback agent $f_i$ will ensure that $\frac{1}{p(N)} \sum_{j=1}^{p(N)} [C_j(x)]_i \approx 0$. Namely, it will ensure that the average of the outputs in dimension $i$ is close to zero. We will show that from any solution $S$ to the \textsc{Consensus-Halving} instance, we obtain a solution to the original \textsc{high-D-Tucker} instance. To be a bit more precise, from the point $x \in [-1,1]^N$ encoded by the Coordinate-Encoding region in a solution $S$, we will be able to extract a polynomial number of points on the \textsc{high-D-Tucker} grid, with the guarantee that two of these points form a solution (which can then be identified efficiently).

\paragraph{\textbf{Encoding of a value in $[-1,1]$.}}
Given any solution $S$ of our instance, every interval $I$ of length $1$ of the domain encodes a value in $[-1,1]$ as follows. Let $\plus$ and $\minus$ denote the subsets of $I$ labeled respectively $+$ and $-$ in the solution $S$. Then the value encoded by $I$, $v_S(I)$, is given by $\mu(\plus) - \mu(\minus)$, where $\mu$ is the Lebesgue measure on $\mathbb{R}$. Since there are at most $n$ cuts (where $n$ is the number of agents in the instance), $\plus$ is the union of at most $n+1$ disjoint sub-intervals of $I$ and $\mu(\plus)$ is simply the sum of the lengths of these intervals (and the same holds for $\minus$). It is easy to see that $v_S(I)=0$ corresponds to $I$ being perfectly shared between $+$ and $-$ in $S$, whereas $v_S(I)=+1$ corresponds to the whole interval $I$ being labeled $+$. We will drop the subscript $S$ and just use $v(I)$ in the remainder of this exposition.

\paragraph{\textbf{Coordinate-Encoding region.}} The sub-interval $[0,N]$ of the domain is called the \emph{Coordinate-Encoding} region. Indeed, the way in which this region is subdivided amongst the $+$ and $-$ labels in a solution $S$ will encode the coordinates of a point in $x \in [-1,1]^N$. In more detail, $x_1 \in [-1,1]$ will be given by $v([0,1])$, i.e.,\ the value encoded by interval $[0,1]$. Similarly, $x_2 \in [-1,1]$ will be given by $v([1,2])$, $x_3 \in [-1,1]$ by $v([2,3])$, etc.

\paragraph{\textbf{Constant-Creation region.}} The sub-interval $[N,N+p(N)]$ of the domain is called the \emph{Constant-Creation} region. This region will be used to create the constants that the circuit-simulators need. The circuit-simulator $C_1$ will read-in the value $v([N,N+1]) =: \textup{const}_1$ and will assume that it corresponds to the value $+1$. Note that given the constant $+1$, the circuit-simulator can create any constant $\zeta \in [-1,1]$ by using a $\times \zeta$-gate (multiplication by the constant $\zeta$). Similarly, the circuit-simulator $C_2$ will read-in the value $v([N+1,N+2]) =: \textup{const}_2$ and use it as the constant $+1$, and so on for $C_3, C_4, \dots, C_{p(N)}$.

If $S$ is a solution such that the Constant-Creation region does not contain any cut, then the whole region will have the same label, and without loss of generality we can assume that this label is $+$. Thus, in such a solution $S$, all the circuit-simulators will indeed read-in the constant $+1$ from the Constant-Creation region, i.e.,\ we will indeed have $\textup{const}_j = +1$ for all $j=1, \dots, p(N)$.

\paragraph{\textbf{Circuit-Simulation regions.}} For each $j \in \{1,2, \dots, p(N)\}$, the sub-interval $[N+p(N)+(j-1)q,N+p(N)+jq]$ of the domain will be used by the circuit-simulator $C_j$. The length $q$ used by every circuit-simulator will be upper-bounded by some polynomial in $N$ and the size of the circuit $\lambda$. Every circuit-simulator $C_j$ will read-in the coordinates $x_1, \dots, x_N \in [-1,1]$ of the point $x$ from the Coordinate-Encoding region, as well as the value $\textup{const}_j \in [-1,1]$ from the Constant-Creation region (and assume that it corresponds to the constant $+1$). Using these values, $C_j$ will perform some computations, including a simulation of the Boolean circuit $\lambda$, and finally output $N$ values $[C_j(x,\textup{const}_j)]_1, \dots, [C_j(x,\textup{const}_j)]_N \in [-1,1]$ into the \emph{Feedback region}.

\paragraph{\textbf{Feedback region.}}
The Feedback region is located at the right end of the domain and is subdivided into $N$ intervals $F_1, \dots, F_N$ of length $p(N)$ each. For every $j \in [p(N)]$, let $F_i(j)$ denote the $j$th sub-interval of length $1$ of $F_i$. The $i$th output of circuit-simulator $C_j$ will be located in sub-interval $F_i(j)$. In other words, $v(F_i(j)) = [C_j(x,\textup{const}_j)]_i$.

Every interval $F_i$ will have a corresponding \emph{feedback agent} $f_i$, who will ensure that the average of all the outputs in interval $F_i$ is close to zero. In more detail, agent $f_i$ will have a single block of value that covers interval $F_i$. As a result, this agent will be satisfied only if $\frac{1}{p(N)}\sum_{j=1}^{p(N)} v(F_i(j)) \in [-\varepsilon, \varepsilon]$.

\paragraph{\textbf{Stray Cuts.}} Any agent belonging to a circuit-simulator performs a gate operation. In \cref{sec:gates}, we introduce the different types of gates and how they are implemented by agents. One important feature of the agents implementing the gates is that every such agent ensures that at least one cut must lie in a specific interval $J$ of the domain (in any solution $S$). By construction, we will make sure that these intervals are pairwise disjoint for different agents. Thus, every agent introduced as part of a circuit-simulator will force one cut to lie in a specific interval.

The only agents that are not part of a circuit-simulator are the feedback agents $f_1, \dots, f_N$. Since the number of cuts in any solution is at most the number of agents, there are at most $N$ cuts that are not constrained to lie in some specific interval. We call these the \emph{free cuts}. The free cuts can theoretically ``go'' anywhere in the domain and interfere with the correct functioning of the circuit-simulators or the Constant-Creation region. The expected behavior of these $N$ free cuts is that they should lie in the Coordinate-Encoding region. As such, any of the free cuts that lies outside the Coordinate-Encoding region will be called a \emph{stray cut} (following \citep{FRG18-Necklace}).

\begin{observation}\label{obs:stray}
If there is at least one stray cut, then the point $x \in [-1,1]^N$ encoded by the Coordinate-Encoding region lies on the boundary of $[-1,1]^N$ (i.e.,\ there exists $i$ such that $|x_i|=1$).
\end{observation}

\paragraph{\textbf{Stray Cut interference.}} There are two ways for a stray cut to cause trouble:
\begin{enumerate}
\item it can \emph{corrupt} a circuit, i.e.,\ interfere with the correct functioning of the gates of a circuit-simulator. If the cut lies in the region of circuit-simulator $C_j$, then it can make a gate output the wrong result (i.e.,\ not perform the desired operation). If the cut lies in the Constant-Creation region and intersects the interval that is used by circuit-simulator $C_i$ to read-in the constant $\textup{const}_j$, then it can have an effect such that $|\textup{const}_j| \neq 1$. However, in any case, a single stray cut can only interfere with one circuit-simulator in this way. Thus, at most $N$ circuit-simulators can suffer from this kind of interference. We will choose $p(N)$ large enough so that these corrupted circuit-simulators have a very limited influence.
\item it can interfere with the sign of $\textup{const}_j$ for many circuit-simulators $C_j$. Indeed, even a single stray cut can ensure that half of our circuit-simulators read-in the constant $+1$ and the other half read-in the constant $-1$. We will show that this is actually not a problem, and that it does not produce bogus solutions. Since stray cuts can only occur when $x$ lies on the boundary of $[-1,1]^N$ (\cref{obs:stray}), the Tucker boundary conditions will be important for this.
\end{enumerate}
Stray cuts that end up in the Feedback region do not have any effect. Indeed, the feedback agents $f_1, \dots, f_N$ are immune to stray cuts. They always ensure that the average of the outputs is close to zero. Thus, a stray cut can only influence the outputs that a feedback agent sees (as detailed above), but not its functionality.

\paragraph{\textbf{Circuit-Simulator failure.}} There are two ways in which a circuit-simulator can fail to have the desired output:
\begin{enumerate}
\item it is corrupted by a stray cut. This can happen to at most $N$ circuit-simulators.
\item it can fail in extracting the binary bits from (a point close to) $x$. We will ensure that this can happen to at most $N$ circuit-simulators.
\end{enumerate}
Thus, at most $2N$ circuit-simulators fail, i.e.,\ at least $p(N)-2N$ circuit-simulators have the desired output.

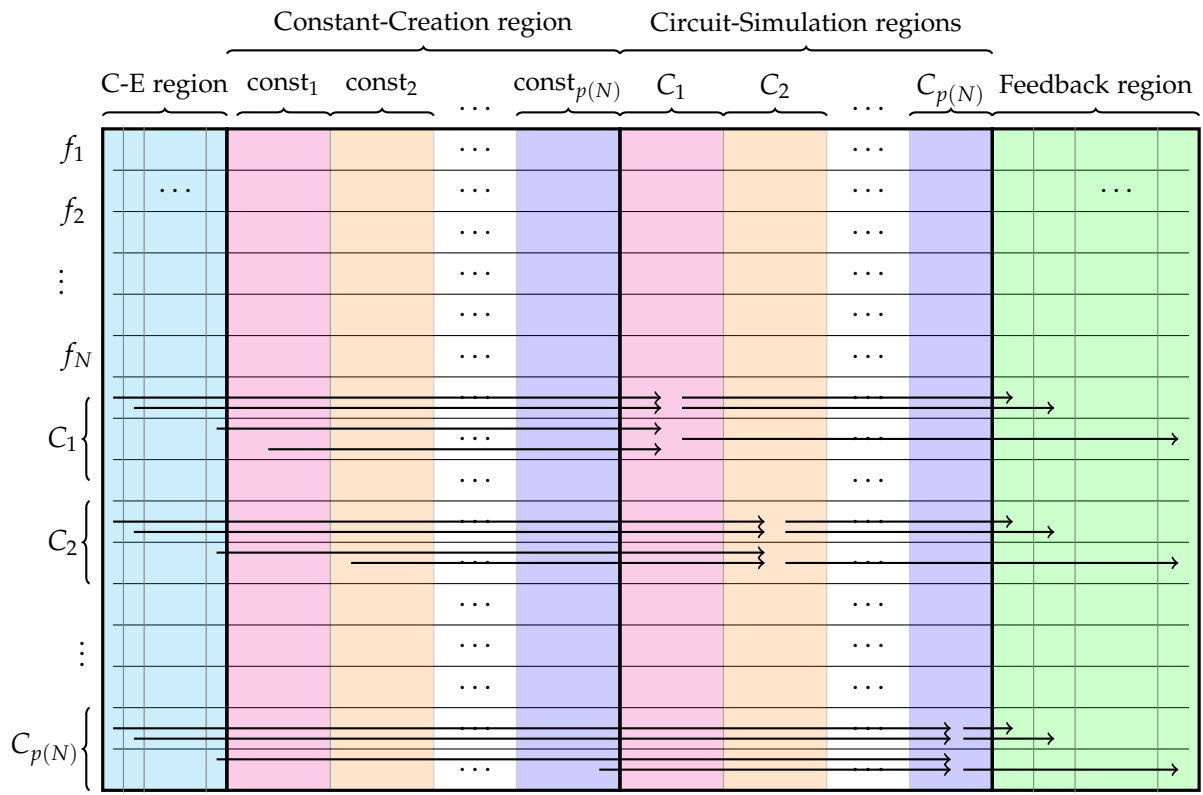
\begin{figure}
    \centering
    \scalebox{0.85}{\begin{tikzpicture}[scale=0.5]
\tikzstyle{myedgestyle} = [cyan]

\node (v1) at (-22,11.5) {};
\node (v2) at (4.5,11.5) {};
\draw  (v1) edge (v2);
\node (v3) at (-22,10.5) {};
\node (v4) at (4.5,10.5) {};
\draw  (v3) edge (v4);
\node (v5) at (-22,9.5) {};
\node (v6) at (4.5,9.5) {};
\draw  (v5) edge (v6);
\node (v7) at (-22,8.5) {};
\node (v8) at (4.5,8.5) {};
\draw  (v7) edge (v8);
\node (v9) at (-22,7.5) {};
\node (v11) at (-22,6.5) {};
\node (v13) at (-22,5.5) {};
\node (v15) at (-22,4.5) {};
\node (v17) at (-22,3.5) {};
\node (v20) at (-22,2.5) {};
\node (v21) at (-22,1.5) {};
\node (v23) at (-22,0.5) {};
\node (v25) at (-22,-0.5) {};
\node (v27) at (-22,-1.5) {};
\node (v29) at (-22,-2.5) {};
\node (v10) at (4.5,7.5) {};
\node (v12) at (4.5,6.5) {};
\node (v14) at (4.5,5.5) {};
\node (v16) at (4.5,4.5) {};
\node (v18) at (4.5,3.5) {};
\node (v19) at (4.5,2.5) {};
\node (v22) at (4.5,1.5) {};
\node (v24) at (4.5,0.5) {};
\node (v26) at (4.5,-0.5) {};
\node (v28) at (4.5,-1.5) {};
\node (v30) at (4.5,-2.5) {};
\draw  (v9) edge (v10);
\draw  (v11) edge (v12);
\draw  (v13) edge (v14);
\draw  (v15) edge (v16);
\draw  (v17) edge (v18);
\draw  (v19) edge (v20);
\draw  (v21) edge (v22);
\draw  (v23) edge (v24);
\draw  (v25) edge (v26);
\draw  (v27) edge (v28);
\draw  (v29) edge (v30);
\draw [fill=cyan, opacity=0.2] (-22,12.5) node (v35) {} rectangle (-19,-3.5) node (v36) {};
\draw [fill=magenta, opacity=0.2] (-19,12.5) node (v31) {} rectangle (-16.5,-3.5);
\draw [fill=orange , opacity=0.2] (-16.5,12.5) rectangle (-14,-3.5);
\draw [fill=blue, opacity=0.2] (-12,12.5) rectangle (-9.5,-3.5) node (v32) {};
\draw [fill=magenta, opacity=0.2] (-9.5,12.5) node (v33) {} rectangle (-7,-3.5);
\draw [fill=orange , opacity=0.2] (-7,12.5) rectangle (-4.5,-3.5);
\draw [fill=blue, opacity=0.2] (-2.5,12.5) rectangle (-0.5,-3.5);
\draw [fill=green,opacity=0.2] (-0.5,12.5) node (v34) {} rectangle (4.5,-3.5);
\draw [black,very thick] (v31) rectangle (v32);
\draw [black,very thick] (v33) rectangle (-0.5,-3.5);
\draw [black,very thick] (v34) rectangle (4.5,-3.5);
\draw [black,very thick](v35) rectangle (v36);
\node (v37) at (-21.5,12.8) {};
\node (v38) at (-21.5,-3.8) {};
\draw [color=gray] (v37) edge (v38);
\node (v39) at (-21,12.8) {};
\node (v40) at (-21,-3.8) {};
\draw [color=gray] (v39) edge (v40);
\node (v41) at (-19.5,12.8) {};
\node (v42) at (-19.5,-3.8) {};
\draw [color=gray] (v41) edge (v42);
\node (v43) at (0.5,12.8) {};
\node (v44) at (0.5,-3.8) {};
\node (v45) at (1.5,12.8) {};
\node (v46) at (1.5,-3.8) {};
\node (v47) at (3.5,12.8) {};
\node (v48) at (3.5,-3.8) {};
\draw [color=gray] (v43) edge (v44);
\draw [color=gray] (v45) edge (v46);
\draw [color=gray] (v47) edge (v48);

\node (a) at (-13,13) {$\ldots$};
\node (a) at (-13,12) {$\ldots$};
\node (a) at (-13,11) {$\ldots$};
\node (a) at (-13,10) {$\ldots$};
\node (a) at (-13,9) {$\ldots$};
\node (a) at (-13,8) {$\ldots$};
\node (a) at (-13,7) {$\ldots$};
\node (a) at (-13,6) {$\ldots$};
\node (a) at (-13,5) {$\ldots$};
\node (a) at (-13,4) {$\ldots$};
\node (a) at (-13,3) {$\ldots$};
\node (a) at (-13,2) {$\ldots$};
\node (a) at (-13,1) {$\ldots$};
\node (a) at (-13,0) {$\ldots$};
\node (a) at (-13,-1) {$\ldots$};
\node (a) at (-13,-2) {$\ldots$};
\node (a) at (-13,-3) {$\ldots$};

\node (a) at (-3.5,13) {$\ldots$};
\node (a) at (-3.5,12) {$\ldots$};
\node (a) at (-3.5,11) {$\ldots$};
\node (a) at (-3.5,10) {$\ldots$};
\node (a) at (-3.5,9) {$\ldots$};
\node (a) at (-3.5,8) {$\ldots$};
\node (a) at (-3.5,7) {$\ldots$};
\node (a) at (-3.5,6) {$\ldots$};
\node (a) at (-3.5,5) {$\ldots$};
\node (a) at (-3.5,4) {$\ldots$};
\node (a) at (-3.5,3) {$\ldots$};
\node (a) at (-3.5,2) {$\ldots$};
\node (a) at (-3.5,1) {$\ldots$};
\node (a) at (-3.5,0) {$\ldots$};
\node (a) at (-3.5,-1) {$\ldots$};
\node (a) at (-3.5,-2) {$\ldots$};
\node (a) at (-3.5,-3) {$\ldots$};

\draw [->,black, thick] (-21.75,6) to (-8.5,6);
\draw [->,black, thick] (-21.25,5.75) to (-8.5,5.75);
\draw [->,black, thick] (-19.25,5.25) to (-8.5,5.25);

\draw [->,black, thick] (-18,4.75) to (-8.5,4.75);

\draw [->,black, thick] (-8,6) to (0,6);
\draw [->,black, thick] (-8,5.75) to (1,5.75);
\draw [->,black, thick] (-8,5) to (4,5);

\draw [->,black, thick] (-21.75,3) to (-6,3);
\draw [->,black, thick] (-21.25,2.75) to (-6,2.75);
\draw [->,black, thick] (-19.25,2.25) to (-6,2.25);

\draw [->,black, thick] (-16,2) to (-6,2);

\draw [->,black, thick] (-5.5,3) to (0,3);
\draw [->,black, thick] (-5.5,2.75) to (1,2.75);
\draw [->,black, thick] (-5.5,2) to (4,2);

\draw [->,black, thick] (-21.75,-2) to (-1.5,-2);
\draw [->,black, thick] (-21.25,-2.25) to (-1.5,-2.25);
\draw [->,black, thick] (-19.25,-2.75) to (-1.5,-2.75);

\draw [->,black, thick] (-10,-3) to (-1.5,-3);

\draw [->,black, thick] (-1.2,-2) to (0,-2);
\draw [->,black, thick] (-1.2,-2.25) to (1,-2.25);
\draw [->,black, thick] (-1.2,-3) to (4,-3);

\draw [
thick,
decoration={
    brace,
    raise=5
},
decorate
] (-22,12.5) -- (-19,12.5)
node [pos=0.5,anchor=north,yshift=25] {\small{C-E region}}; 

\draw [
thick,
decoration={
    brace,
    raise=5
},
decorate
] (-19,14) -- (-9.5,14)
node [pos=0.5,anchor=north,yshift=25] {\small{Constant-Creation region}}; 

\draw [
thick,
decoration={
    brace,
    raise=5
},
decorate
] (-9.5,14) -- (-0.5,14)
node [pos=0.5,anchor=north,yshift=25] {\small{Circuit-Simulation regions}};

\draw [
thick,
decoration={
    brace,
    raise=5
},
decorate
] (v31) -- (-16.5,12.5)
node [pos=0.5,anchor=north,yshift=25] {\small{$\text{const}_1$}}; 

\draw [
thick,
decoration={
    brace,
    raise=5
},
decorate
] (-16.5,12.5) -- (-14,12.5)
node [pos=0.5,anchor=north,yshift=25] {\small{$\text{const}_2$}};

\draw [
thick,
decoration={
    brace,
    raise=5
},
decorate
] (-12,12.5) -- (-9.5,12.5)
node [pos=0.5,anchor=north,yshift=25] {\small{$\text{const}_{p(N)}$}};

\draw [
thick,
decoration={
    brace,
    raise=5
},
decorate
] (-9.5,12.5) -- (-7,12.5)
node [pos=0.5,anchor=north,yshift=25] {$C_1$}; 

\draw [
thick,
decoration={
    brace,
    raise=5
},
decorate
] (-7,12.5) -- (-4.5,12.5)
node [pos=0.5,anchor=north,yshift=25] {$C_2$};

\draw [
thick,
decoration={
    brace,
    raise=5
},
decorate
] (-2.5,12.5) -- (-0.5,12.5)
node [pos=0.5,anchor=north,yshift=25] {$C_{p(N)}$}; 

\draw [
thick,
decoration={
    brace,
    raise=5
},
decorate
] (-0.5,12.5) -- (4.5,12.5)
node [pos=0.5,anchor=north,yshift=25] {\small{Feedback region}};

\draw [
thick,
decoration={
    mirror,
    raise=5
},
decorate
] (-22,12.5) -- (-22,11.5)
node [pos=0.5,anchor=west,xshift=-20] {$f_1$};

\draw [
thick,
decoration={
    mirror,
    raise=5
},
decorate
] (-22,11) -- (-22,10)
node [pos=0.5,anchor=west,xshift=-20] {$f_2$};

\node (dot) at (-23, 9) {$\vdots$};

\draw [
thick,
decoration={
    mirror,
    raise=5
},
decorate
] (-22,7.5) -- (-22,6.5)
node [pos=0.5,anchor=west,xshift=-20] {$f_N$};

\draw [
thick,
decoration={
    brace,
    mirror,
    raise=5
},
decorate
] (-22,6) -- (-22,4)
node [pos=0.5,anchor=west,xshift=-25] {$C_1$};

\draw [
thick,
decoration={
    brace,
    mirror,
    raise=5
},
decorate
] (-22,3.5) -- (-22,1.5)
node [pos=0.5,anchor=west,xshift=-25] {$C_2$};

\node (dot) at (-22.5, 0) {$\vdots$};

\draw [
thick,
decoration={
    brace,
    mirror,
    raise=5
},
decorate
] (-22,-1.5) -- (-22,-3.5)
node [pos=0.5,anchor=west,xshift=-39] {$C_{p(N)}$};

\node (dot) at (-20.25, 11) {$\ldots$};
\node (dot) at (2.5, 11) {$\ldots$};

\end{tikzpicture} }
    \caption{An overview of the different regions defined in the reduction. The regions corresponding to different Circuit-Simulators are color-coded. An arrow indicates that the region where it is pointing receives inputs from the region from where it is originating. On the left, the different types of agents are shown, namely the feedback agents, as well as the agents corresponding to the different Circuit-Simulators. The Coordinate-Encoding region and the Feedback Region are divided into sub-intervals, indicated by vertical gray lines, as detailed in \cref{sec:overview}.}
    \label{fig:overview}
\end{figure}

\subsection{Major simplifications compared to the previous proof}

The major simplifications compared to the previous proof of \citet{FRG18-Necklace} are as follows:

\paragraph{\textbf{A much cleaner domain.}} The PPA-hardness of \textsc{high-D-Tucker} was already established in \citep{FRG18-Necklace} and our version can be obtained from that one using minor modifications, see \cref{thm:tuckerppa}. The corresponding result of \citep{FRG18-Necklace} is a standard application of the ``snake-embedding'' technique developed in \citep{chen2009settling}. However, the reduction in \citep{FRG18-Necklace} requires (a) a further constraint on how the domain is colored and more importantly (b) the embedding of the \textsc{high-D-Tucker} instance into a M{\"o}bius-type simplex domain, in which two facets have been ``identified'' with each other - one can envision a high-dimensional M{\"o}bius strip with an instance of \textsc{high-D-Tucker} in its center, embedding in a high-dimensional simplex.  A key step in the reduction is the extension of the labeling of \textsc{high-D-Tucker} to the remainder of the domain, in a way that does not introduce any artificial solutions, and such that solutions to \textsc{high-D-Tucker} can be traced back from solutions on other points on the domain. For this purpose, the authors of \citep{FRG18-Necklace} develop a rather complicated coordinate transformation, applied to the inputs read from the positions of the cuts. They establish how to compute the transformation and its inverse in polynomial time and how distances in the two coordinate systems (before and after the transformation) are polynomially related.
In contrast, our reduction works with the rather clean domain of \textsc{high-D-Tucker}, avoiding all the unnecessary technical clutter of the domain used in \citep{FRG18-Necklace}.

\paragraph{\textbf{Simpler gadgetry.}} Another complication of the proof in \citep{FRG18-Necklace} is the use of blanket-sensor agents, which constrain the positions of the cuts in the coordinate-encoding region, to ensure that solutions to $\varepsilon$-\ch\ do not encode points that lie too far from a specific region in the ``middle'' of the domain, called the ``significant region''; this is achieved via appropriate feedback provided by these agents to the coordinate-encoding agents. To make sure that the blanket-sensor agents do not ``cancel'' each other, extra care must be taken on how the feedback of these agents is designed, giving rise to a series of technical lemmas. Our reduction does not need to use any such agents and is therefore significantly simpler in that regard as well. 

\paragraph{\textbf{Label sequence robustness.}} The reduction in \citep{FRG18-Necklace} requires knowledge of the label sequence, i.e., whether the first cut that occurs in the c-e region has the label $+$ or $-$ on its left side. This is fundamental for the design of the gates, as they read the inputs as the distances from the left endpoints of the corresponding designated intervals, unlike our interpretation, which measures the difference between the value of the two labels. Thus, for the disorientation of the domain and to deal with sign flips that happen due to the stray cuts, the authors of \citep{FRG18-Necklace} employ a pre-processing circuit that uses the first coordinate-detecting agent as a reference agent, when performing computations. This is again not needed in our case; our equivariant gates ensure that even when the corresponding point lies on the boundary of the \textsc{high-D-Tucker} domain, the output is computed correctly in a much simpler way.

\subsection{Arithmetic Gates}\label{sec:gates}

In this section we show how to construct gates which perform various operations on numbers in $[-1,1]$ with error at most $g(\varepsilon) = 16 \varepsilon$, where $\varepsilon$ is the error we allow in a \textsc{Consensus-Halving} solution. Some of these gates will be immune to ``corruption'' by a stray cut, while others might get corrupted and not work properly.

Recall that in any solution $S$, any unit length interval $I$ of the domain represents value $v(I) \in [-1,1]$. We will now show how to perform computations with these values. We let $T : \mathbb{R} \to [-1,1]$, $z \mapsto \max(-1, \min(1,z))$, i.e.,\ $T[z]$ is the \emph{truncation} of $z \in \mathbb{R}$ in $[-1,1]$. We will also abuse notation and use $T[x] = (T[x_1], \dots, T[x_N])$ for $x \in [-1,1]^N$. At this stage, we assume that $\varepsilon$ is sufficiently small for the gates to work (namely $\varepsilon \leq 2^{-10}$ is enough). 

We will design \emph{basic gates}, namely \textbf{\emph{Multiplication by $-1$ $[G_{\times (-1)}]$}}, \textbf{\emph{Constant $\zeta \in [-1,1] \cap \mathbb{Q}$ $[G_{\zeta}]$}} and \emph{\textbf{Addition $[G_{+}]$}}, and \emph{additional gates}, namely \emph{\textbf{Copy $[G_{\textup{copy}}]$}}, \emph{\textbf{Multiplication by $k \in \mathbb{N}$ $[G_{\times k}]$}} and \textbf\emph{Boolean Gates}, \emph{\textbf{Negation $[G_{\lnot}]$}}, \emph{\textbf{AND $[G_{\land}]$}} and \emph{\textbf{OR $[G_{\lor}]$}}. 

\subsubsection{Basic Gates}

\paragraph{\textbf{$\delta$-Volume Gate $[G_{\delta}]$:}} Let $\delta \in [2\varepsilon,1]$. Let $I$ and $O$ be disjoint intervals of length $1$. The agent for this gate has a block of length $1-\delta$ and height $\frac{1}{2-\delta}$ centered in interval $I$ and a block of length $1$ and (same) height $\frac{1}{2-\delta}$ in interval $O$. \smallskip

\noindent It is easy to check that since $\delta \geq 2\varepsilon$, at least one cut must lie strictly within $O$ in any solution. Furthermore, this gate has the notable property that it cannot be corrupted. From this construction, we obtain the following gates:
\begin{itemize}
    \item[-] \textbf{\emph{Multiplication by $-1$ $[G_{\times (-1)}]$:}} Set $\delta = 2\varepsilon$. Then, in any solution it holds that $v(O) = T[-v(I) \pm 4 \varepsilon]$. See \cref{fig:mulbyminus1} for an illustration. \smallskip
    \item[-] \textbf{\emph{Constant $\zeta \in [-1,1] \cap \mathbb{Q}$ $[G_{\zeta}]$:}} To create a constant, we use an input interval $I$ in the Constant-Creation region. Let $j$ be such that this gate is part of circuit-simulator $C_j$. In this case, we use input $I:=[N+(j-1),N+j]$ (the region corresponding to $\textup{const}_j$). If $v(I)=+1$ (i.e.,\ $\textup{const}_j = +1$), then this gate will work properly and cannot be corrupted. 
    \begin{itemize}
    \item For $\zeta \leq 0$, we let $\delta=\max(1+\zeta,2\varepsilon)$ and obtain that $v(O) = T[\zeta \pm 4 \varepsilon]$. 
    \item For $\zeta > 0$, use $G_{-\zeta}$ and then $G_{\times (-1)}$, for a total error of at most $8 \varepsilon$. Note that if $|v(I)| = 1$, then this gate can also be used to obtain $T[v(I) \times \zeta \pm 8 \varepsilon]$.
    \end{itemize} See \cref{fig:constant} for an illustration.
\end{itemize}

\begin{figure}
\centering
\begin{tikzpicture}[scale=1,transform shape]
	\node (a_1) at (-10pt,0pt) {\small{$G_{\times(-1)}$}}; 
	\node (a_2) at (300pt, 0pt) {};
	\draw (a_1)--(a_2);
	\draw[fill=red!20!white] (24pt,0pt) rectangle (56pt,22pt);

	\draw[fill=red!20!white] (230pt,0pt) rectangle (270pt,22pt);
	
		\draw[fill=gray!20!white] (20pt,0pt) rectangle (24pt,22pt);
	\draw[fill=gray!20!white] (56pt,0pt) rectangle (60pt,22pt);

	\node (a_1) at (-10pt,-50pt) {\small{$G_{\times(-1)}$}}; 
	\node (a_2) at (300pt, -50pt) {};
	\draw (a_1)--(a_2);
	\draw[fill=red!20!white] (24pt,-50pt) rectangle (56pt,-28pt);
	\draw[fill=gray!20!white] (20pt,-50pt) rectangle (24pt,-28pt);
	\draw[fill=gray!20!white] (56pt,-50pt) rectangle (60pt,-28pt);

	\draw[fill=red!20!white] (230pt,-50pt) rectangle (270pt,-28pt);
	
		\draw [
    thick,
    decoration={
        brace,
        mirror,
        raise=5pt
    },
    decorate
] (20pt,-50pt) -- (24pt,-50pt)
node [pos=0.5,anchor=north,yshift=-8pt] {$2\varepsilon$};

	\draw [
    thick,
    decoration={
        brace,
        mirror,
        raise=5pt
    },
    decorate
] (20pt,-70pt) -- (60pt,-70pt)
node [pos=0.5,anchor=north,yshift=-10pt] {\small{${I}$}};

\draw [
    thick,
    decoration={
        brace,
        mirror,
        raise=5pt
    },
    decorate
] (230pt, -70pt) -- (270pt, -70pt)
node [pos=0.5,anchor=north,yshift=-10pt] {\small{$O$}};

\draw[dashed,color=blue] (258pt,50pt) -- (258pt, -15pt);

\draw[dashed,color=blue] (32pt,50pt) -- (32pt, -70pt);
\draw[dashed,color=red] (242pt,-10pt) -- (242pt, -70pt);

\node[color=blue] at (20pt,30pt) {$+$};
\node[color=blue] at (40pt,30pt) {$-$};

\node[color=blue] at (230pt,30pt) {$+$};
\node[color=blue] at (270pt,30pt) {$-$};

\node[color=red] at (230pt,-20pt) {$-$};
\node[color=red] at (270pt,-20pt) {$+$};

\end{tikzpicture} \caption{A  \textbf{\emph{Multiplication by $-1$ $[G_{\times (-1)}]$}} gate. The gray shaded regions are not part of the agents valuation on $I$, but are shown for clarity. The input to the gate is $-1/2$ and the output is $1/2 \pm 4\varepsilon$. Two different sets of cut are shown (top and bottom), to emphasize the fact that the gadgets are resilient to flips in the parity of the cut sequence. On both cases, the parity sequence on the left is $+/-$; on the top, the parity sequence on the right is also $+/-$, and the cut is placed in the rightmost half in $O$, whereas on the bottom, the parity sequence on the right is $-/+$, and the cut is placed in the leftmost half in $O$.}
\vspace{0.5cm}
\label{fig:mulbyminus1}
\end{figure}
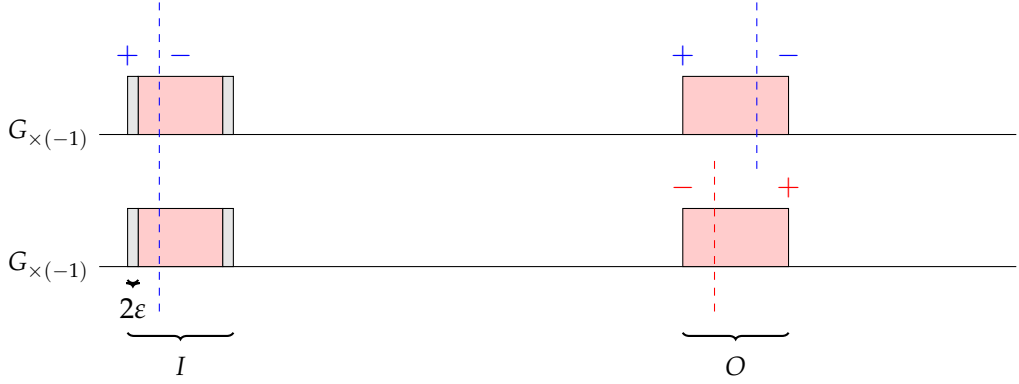

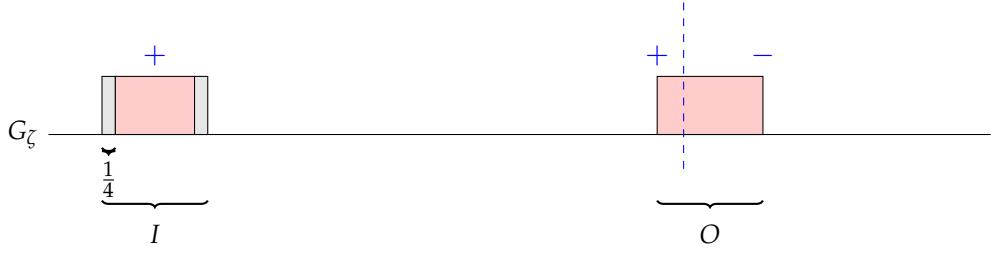
\begin{figure}
\centering
\begin{tikzpicture}[scale=1,transform shape]
	\node (a_1) at (-10pt,0pt) {\small{$G_{\zeta}$}}; 
	\node (a_2) at (300pt, 0pt) {};
	\draw (a_1)--(a_2);
	\draw[fill=red!20!white] (25pt,0pt) rectangle (55pt,22pt);

	\draw[fill=red!20!white] (230pt,0pt) rectangle (270pt,22pt);
	
		\draw[fill=gray!20!white] (20pt,0pt) rectangle (25pt,22pt);
	\draw[fill=gray!20!white] (55pt,0pt) rectangle (60pt,22pt);

		\draw [
    thick,
    decoration={
        brace,
        mirror,
        raise=5pt
    },
    decorate
] (20pt,0pt) -- (25pt,0pt)
node [pos=0.5,anchor=north,yshift=-6pt] {$\frac{1}{4}$};

	\draw [
    thick,
    decoration={
        brace,
        mirror,
        raise=5pt
    },
    decorate
] (20pt,-20pt) -- (60pt,-20pt)
node [pos=0.5,anchor=north,yshift=-10pt] {\small{${I}$}};

\draw [
    thick,
    decoration={
        brace,
        mirror,
        raise=5pt
    },
    decorate
] (230pt, -20pt) -- (270pt, -20pt)
node [pos=0.5,anchor=north,yshift=-10pt] {\small{$O$}};

\draw[dashed,color=blue] (240pt,50pt) -- (240pt, -15pt);

\node[color=blue] at (40pt,30pt) {$+$};

\node[color=blue] at (230pt,30pt) {$+$};
\node[color=blue] at (270pt,30pt) {$-$};

\end{tikzpicture} \caption{A  \textbf{\emph{Constant $\zeta \in [-1,1] \cap \mathbb{Q}$ $[G_{\zeta}]$}} gate. The gray shaded regions are not part of the agents valuation on $I$, but are shown for clarity. $I$ is part of the Constant-Creation region and therefore (in a well-behaved case), it is not intersected by any cuts. The value block of the agent on the left is labeled entirely by ``$+$'', and the cut on the right side assumes the corresponding position in favor of ``$-$''m to balance out the discrepancy. In the example, $\zeta = -3/4$ and therefore $\delta = 1/4$, and the output on the right is $-3/4$.}
\label{fig:constant}
\end{figure}

\paragraph{\textbf{Addition $[G_{+}]$:}} Let $I_1$ and $I_2$ be the two length-$1$ intervals encoding the two inputs. Let $I'$ be an interval of length $2$ that is disjoint from $I_1$ and $I_2$. Let $J$ be an interval of length $3$ that is disjoint from $I_1, I_2$ and $I'$. We first use a $G_{\times (-1)}$-gate with input $I_1$ and output $I'[0,1]$, and another one with input $I_2$ and output $I'[1,2]$.\smallskip

\noindent To compute addition we create a new agent with valuation function that has height $1/5$ in $I'$ and in $J$, and height $0$ everywhere else. Note that since $\varepsilon < 1/5$, in any solution there must be a cut in $J$. We say that the gate is corrupted, if there are at least two cuts lying strictly in the interval $J$. The output of the gate will be in interval $O = J[1,2]$. If the gate is not corrupted, then by construction we have $v(O) = T[v(I_1) + v(I_2) \pm 16 \varepsilon]$. See \cref{fig:addition} for an illustration.

\begin{figure}
\centering
\begin{tikzpicture}[scale=0.8,transform shape]
	\node (a_1) at (-10pt,0pt) {\small{$G_{\times(-1)}$}}; 
	\node (a_2) at (430pt, 0pt) {};
	\draw (a_1)--(a_2);
	
	\node (a_1) at (-10pt,-50pt) {\small{$G_{\times(-1)}$}}; 
	\node (a_2) at (430pt, -50pt) {};
	\draw (a_1)--(a_2);
	
	\node (a_1) at (-10pt,-100pt) {\small{$G_{+}$}}; 
	\node (a_2) at (430pt, -100pt) {};
	\draw (a_1)--(a_2);

	\draw[fill=red!20!white] (24pt,0pt) rectangle (56pt,22pt);
	\draw[fill=gray!20!white] (20pt,0pt) rectangle (24pt,22pt);
	\draw[fill=gray!20!white] (56pt,0pt) rectangle (60pt,22pt);
	
	\node (dots1) at (70pt,-25pt) {$\ldots$};
	
	\draw[fill=red!20!white] (94pt,-50pt) rectangle (126pt,-28pt);
	\draw[fill=gray!20!white] (90pt,-50pt) rectangle (94pt,-28pt);
	\draw[fill=gray!20!white] (126pt,-50pt) rectangle (130pt,-28pt);
	
	\node (dots1) at (150pt,-25pt) {$\ldots$};

	\draw[fill=red!20!white] (160pt,0pt) rectangle (200pt,22pt);

	\draw[fill=red!20!white] (200pt,-50pt) rectangle (240pt,-28pt);
	
	\draw[fill=red!20!white] (160pt,-100pt) rectangle (240pt,-92pt);
	
	\node (dots1) at (270pt,-75pt) {$\ldots$};
	
	\draw[fill=red!20!white] (300pt,-100pt) rectangle (420pt,-92pt);

		\draw [
    thick,
    decoration={
        brace,
        mirror,
        raise=5pt
    },
    decorate
] (20pt,0pt) -- (24pt,0pt)
node [pos=0.5,anchor=north,yshift=-8pt] {$2\varepsilon$};

	\draw [
    thick,
    decoration={
        brace,
        mirror,
        raise=5pt
    },
    decorate
] (20pt,-100pt) -- (60pt,-100pt)
node [pos=0.5,anchor=north,yshift=-10pt] {\small{${I_1}$}}; 

	\draw [
    thick,
    decoration={
        brace,
        mirror,
        raise=5pt
    },
    decorate
] (90pt,-100pt) -- (130pt,-100pt)
node [pos=0.5,anchor=north,yshift=-10pt] {\small{${I_2}$}};

\draw [
    thick,
    decoration={
        brace,
        mirror,
        raise=5pt
    },
    decorate
] (160pt, -100pt) -- (240pt, -100pt)
node [pos=0.5,anchor=north,yshift=-10pt] {\small{$I'$}};

\draw [
    thick,
    decoration={
        brace,
        mirror,
        raise=5pt
    },
    decorate
] (340pt, -100pt) -- (380pt, -100pt)
node [pos=0.5,anchor=north,yshift=-10pt] {\small{$O$}};

\draw [
    thick,
    decoration={
        brace,
        mirror,
        raise=5pt
    },
    decorate
] (300pt, -120pt) -- (420pt, -120pt)
node [pos=0.5,anchor=north,yshift=-10pt] {\small{$J$}};

\draw[dashed,color=blue] (32pt,35pt) -- (32pt, -110pt);
\draw[dashed,color=blue] (98pt,35pt) -- (98pt, -110pt);

\draw[dashed,color=blue] (192pt,35pt) -- (192pt, -110pt);
\draw[dashed,color=blue] (204pt,35pt) -- (204pt, -110pt);

\draw[dashed,color=blue] (400pt,35pt) -- (400pt, -110pt);

\node[color=blue] at (20pt,40pt) {$+$};
\node[color=blue] at (40pt,40pt) {$-$};

\node[color=blue] at (90pt,40pt) {$+$};
\node[color=blue] at (108pt,40pt) {$-$};

\node[color=blue] at (180pt,40pt) {$+$};
\node[color=blue] at (198pt,40pt) {$-$};
\node[color=blue] at (218pt,40pt) {$+$};

\node[color=blue] at (390pt,40pt) {$-$};
\node[color=blue] at (410pt,40pt) {$+$};

\end{tikzpicture} \caption{An \emph{\textbf{Addition $[G_{+}]$}} gate. An arbitrary sequence of labels is shown. First, we use two $G_{\times(-1)}$ gates for the two inputs and the outputs of those gates are ``read together'' in a single interval $I'$. The output of the $G_{+}$ gate is read from $O$. In the example, the first input is $-1/2$ and the second input is $-3/4$. The cut in $J$ is placed appropriately to balance out the excess of ``$+$'' in $I'$, but the output of the gate is actually $-1$, i.e., the gate applies truncation to the output of the addition as it is smaller than $-1$.}
\label{fig:addition}
\end{figure}
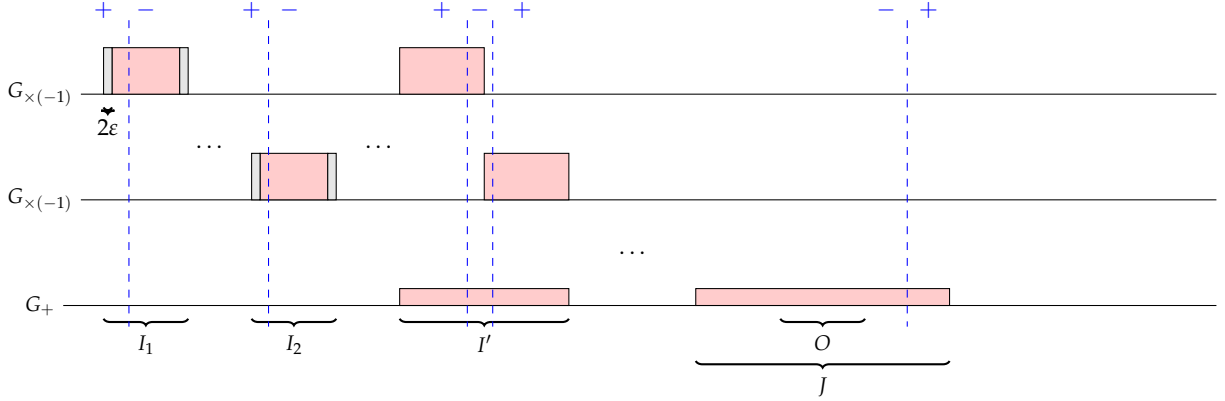

\subsubsection{Additional Gates}

Using the gates presented above, we can also implement the following operations.

\paragraph{\textbf{Copy $[G_{\textup{copy}}]$:}} We can copy a value by using two successive $G_{\times (-1)}$-gates. The error will be at most $8 \varepsilon$. This gate can not be corrupted.

\paragraph{\textbf{Multiplication by $k \in \mathbb{N}$ $[G_{\times k}]$:}} This can be implemented by a chain of $k-1$ additions that repeatedly add the input to the intermediate summation. It is easy to check that the error is at most $(k-1)16 \varepsilon$.

\paragraph{\textbf{Boolean Gates:}} For the Boolean gates, the bits $\{0,1\}$ will be represented by the values $-1$ (for $0$) and $+1$ (for $1$). We say that an input value $v(I) \in [-1,1]$ is a \emph{perfect bit}, if $v(I) \in \{-1,1\}$. The Boolean gates we will construct have the following very nice property: if all inputs are perfect bits, then the output is a perfect bit (and is the correct output).
\begin{itemize}
    \item[-] \emph{\textbf{Negation Gate $[G_{\lnot}]$:}}. This can be done by first using a $G_{\times (-1)}$-gate and then a $G_{\times 2}$-gate. The error will be at most $20\varepsilon$, so this will work as intended as long as $20\varepsilon \leq 1$.
    \item[-] \emph{\textbf{AND Gate: $[G_{\land}]$:}} Let $b_1$ and $b_2$ be the two inputs. We perform the computation $((b_1+b_2)-1/2) \times 4$ by using two $G_+$-gates, a $G_{-1/2}$-gate and a $G_{\times 4}$-gate. The error will be at most $12 \times 16 \varepsilon$, so the gate will work as intended as long as $12 \times 16 \varepsilon \leq 1$.
    \item[-] \emph{\textbf{OR Gate $[G_{\lor}]$:}} This can easily be obtained by using $G_{\land}$ and $G_{\lnot}$.
\end{itemize}
\noindent Note that all the arithmetic gates we have presented have error at most $g(\varepsilon)=16\varepsilon$, with the exception of $G_{\times k}$, which has an error of $(k-1)g(\varepsilon)$. Thus, we have to be careful whenever we use this gate for some non-constant $k$.

\begin{remark}[\textbf{Equivariant Gates}]\label{rem:equivariant-gates}
Note that the operation performed by any gate is \emph{equivariant}. Namely, if we flip the sign of all inputs, the same output is still valid, but with a flipped sign. For $G_{\times -1}$ and $G_{+}$ gates this is obvious. For $G_{\zeta}$, we have to recall that $\textup{const}_j$ is the input to the gate. With this interpretation, the equivariance is once again obvious. For the $G_{\land}$-gate the equivariance is a bit more subtle. Note that it uses a $G_{-1/2}$-gate, which uses $\textup{const}_j$. Thus, in this case again, if we flip the sign of the inputs $b_1$, $b_2$ and $\textup{const}_j$, the sign of the output is flipped too.

This property of the gates is not a coincidence. It follows from the way we are encoding values. Note that in any solution $S$, if we swap the labels $+$ and $-$, the solution remains valid. The only thing that has changed is that for any gate, the sign of all inputs and outputs has been flipped.

It follows that any circuit that we construct out of these gates will be equivariant. Namely, the computation will still be valid if we flip the sign of all inputs (including $\textup{const}_j$) and all outputs.
\end{remark}

\subsection{Circuit-Simulators}

In this section we describe the functionality of the circuit-simulators and what it achieves. Recall that every circuit-simulator $C_j$ reads-in inputs $x_1, \dots, x_N \in [-1,1]$ from the Coordinate-Encoding region and $\textup{const}_j \in [-1,1]$ from the Constant-Creation region. The circuit-simulator then performs some computations and outputs $[C_j(x, \textup{const}_j)]_1, \dots, [C_j(x, \textup{const}_j)]_N \in [-1,1]$ into the Feedback region. If the circuit-simulator $C_j$ is corrupted (i.e.,\ one of the stray cuts interferes with it), then we will not claim anything about the outputs of $C_j$. In that case, we will only use the fact that all the outputs must lie in $[-1,1]$ (which is guaranteed by the way values are represented).

If the circuit-simulator $C_j$ is not corrupted, then we know that all gates will perform correct computations, and we also know that $\textup{const}_j \in \{-1,+1\}$. In our construction of $C_j$, we will be assuming that $\textup{const}_j = +1$. However, we will show later that even if $\textup{const}_j = -1$, $C_j$ will output something useful.

\paragraph{\textbf{Phase 1: equi-angle displacement.}}
In the first phase, $C_j$ applies a small displacement to its input $x$. Namely, for every $i \in [N]$, $C_j$ computes $\widehat{x}_i \approx T[x_i + j \alpha]$, where $\alpha = \frac{1}{16p(N)}$. This is achieved by using a $G_{j \alpha}$-gate to create the constant $j \alpha$ (by using $\textup{const}_j$), followed by a $G_+$-gate to perform the addition $x_i + j\alpha$. However, since $\varepsilon > 0$, the gates might make some error in the computations. Nevertheless, by construction of the gates, we immediately obtain:

\begin{claim}\label{clm:phase1}
Assume that the circuit-simulator $C_j$ is not corrupted and $\textup{const}_j = +1$. Then the equi-angle displacement phase outputs $\widehat{x}_i = T[x_i + j\alpha \pm 2g(\varepsilon)]$ for all $i$.
\end{claim}

\paragraph{\textbf{Phase 2: bit extraction.}}
In the second phase, $C_j$ extracts the three most significant bits from each $\widehat{x}_i \in [-1,1]$. These three bits tell us where $\widehat{x}_i$ lies in $[-1,1]$, namely in which of the eight possible standard intervals of length $1/4$: $[-1,-3/4]$, $[-3/4,-1/2]$, $[-1/2,-1/4]$, $[-1/4,0]$, $[0,1/4]$, $[1/4,1/2]$, $[1/2,3/4]$, $[3/4,1]$. Instead of the usual $\{0,1\}$, our bits will take values in $\{-1,+1\}$. The first bit $b_1 \in \{-1,+1\}$ indicates whether $\widehat{x}_i$ is positive or negative. If  $b_1 = +1$, then $\widehat{x}_i \in [0,1]$. If $b_1 = -1$, then $\widehat{x}_i \in [-1,0]$. The second bit $b_2 \in \{-1,+1\}$, then indicates in which half of that interval $\widehat{x}_i$ lies. Thus, if $b_1=+1$ and $b_2=-1$, then $\widehat{x}_i \in [0,1/2]$. Note that some of the bits are not well-defined if $\widehat{x}_i \in B$ where $B = \{-3/4,-1/2,-1/4,0,1/4,1/2,3/4\}$. Thus, we cannot expect the bit extraction to succeed in this case. In fact, the bit extraction will fail if $\widehat{x}_i$ is sufficiently close to any point in $B$.

The bit extraction for $\widehat{x}_i$ is performed as follows.
\begin{enumerate}
    \item $b_1 \approx T[\widehat{x}_i \times \lceil 1/g(\varepsilon)\rceil]$ (use $G_{\times \lceil 1/g(\varepsilon)\rceil}$-gate)
    \item $\widehat{x}_i' \approx T[\widehat{x}_i - b_1/2]$ (use $G_{-1/2}$-gate and $G_{+}$-gate)
    \item $b_2 \approx T[\widehat{x}_i' \times \lceil 1/g(\varepsilon)\rceil]$
    \item $\widehat{x}_i'' \approx T[\widehat{x}_i' - b_2/4]$
    \item $b_3 \approx T[\widehat{x}_i'' \times \lceil 1/g(\varepsilon)\rceil]$
\end{enumerate}

Note that to compute $-b_1/2$ and $-b_2/4$ we just use the corresponding constant gate, namely $G_{-1/2}$ and $G_{-1/4}$ (with input $b_1$ or $b_2$ respectively, instead of $\textup{const}_j$). The computation may be incorrect if $b_1$ or $b_2$ are not in $\{-1,1\}$, but in that case the bit-extraction has already failed anyway.

We can show that the bit-extraction succeeds if $\widehat{x}_i$ is sufficiently far away from any point in $B$. Letting $\textup{dist}(t, B) = \min_{p \in B} |t-p|$, we obtain:

\begin{claim}\label{clm:phase2}
Assume that the circuit-simulator $C_j$ is not corrupted and $\textup{const}_j = +1$. If $\textup{dist}(T[x_i + j \alpha], B) \geq 8g(\varepsilon)$, then the bit-extraction phase for $\widehat{x}_i$ outputs the correct bits for $T[x_i + j \alpha]$.
\end{claim}

\begin{proof}
Since $\textup{dist}(T[x_i + j \alpha], B) \geq 8 g(\varepsilon)$, it follows by \cref{clm:phase1} that $\textup{dist}(\widehat{x}_i, B) \geq 6g(\varepsilon)$, and in particular $|\widehat{x}_i| \geq 6g(\varepsilon)$. In step 1 we use a $G_{\times \lceil 1/g(\varepsilon)\rceil}$-gate on input $\widehat{x}_i$. By construction of the gate, it follows that $b_1 = T[\widehat{x}_i \times \lceil 1/g(\varepsilon)\rceil \pm (\lceil 1/g(\varepsilon)\rceil-1) g(\varepsilon)] = T[\widehat{x}_i \times \lceil 1/g(\varepsilon)\rceil \pm 1]$ where we used $|(\lceil 1/g(\varepsilon)\rceil-1) g(\varepsilon)| \leq 1$. Since $|\widehat{x}_i| \geq 6g(\varepsilon)$, it holds that $|\widehat{x}_i \times \lceil 1/g(\varepsilon) \rceil| \geq 6 \geq 2$. Thus, $b_1 \in \{-1,+1\}$ is the correct bit for $\widehat{x}_i$.

In step 2 we compute $\widehat{x}_i' = T[\widehat{x}_i - b_1/2 \pm 2g(\varepsilon)]$. Thus, we have $\textup{dist}(\widehat{x}_i', B) \geq 4g(\varepsilon)$, and in particular $|\widehat{x}_i'| \geq 4g(\varepsilon)$,  which implies that $b_2 = T[\widehat{x}_i' \times \lceil 1/g(\varepsilon)\rceil \pm 1] \in \{-1,+1\}$ is the correct first bit for $\widehat{x}_i'$ and the correct second bit for $\widehat{x}_i$.

In step 4 we compute $\widehat{x}_i'' = T[\widehat{x}_i' - b_2/2 \pm 2g(\varepsilon)]$. Thus, $|\widehat{x}_i''| \geq 2g(\varepsilon)$, which implies that $b_3 = T[\widehat{x}_i'' \times \lceil 1/g(\varepsilon)\rceil \pm 1] \in \{-1,+1\}$ is the correct first bit for $\widehat{x}_i''$, i.e.,\ the correct second bit for $\widehat{x}_i'$ and the correct third bit for $\widehat{x}_i$.

Finally, note that if $\textup{dist}(T[x_i + j \alpha], B) \geq 8 g(\varepsilon)$, then $T[x_i + j \alpha]$ and $\widehat{x}_i$ must lie in the same standard interval of length $1/4$, i.e.,\ they must have the same three bits.
\end{proof}

\paragraph{\textbf{Phase 3: simulation of $\lambda$.}}
Recall that $\lambda: [8]^N \to \{\pm 1, \dots, \pm N\}$ is the Boolean circuit computing the \textsc{high-D-Tucker} labeling. We interpret $[8]$ as a subdivision of $[-1,1]$ into standard intervals of length $1/4$. Namely, $1$ corresponds to $[-1,-3/4]$, $2$ to $[-3/4,-1/2]$, etc. Thus, $[8]^N$ can be interpreted as a subdivision of $[-1,1]^N$ into hypercubes of side-length $1/4$. With this in mind, we define $\overline{\lambda} : ([-1,1] \setminus B)^N \to \{\pm 1, \dots, \pm N\}$, so that for any $x \in ([-1,1] \setminus B)^N$, $\overline{\lambda}(x)$ is the label that $\lambda$ assigns to the hypercube containing $x$.

We can assume that the inputs of $\lambda$ consist of three bits each, such that the number represented by these three bits yields an element in $[8]$ (by using $[8] \equiv \{0,1,\dots,7\}$). The three bits $b_1,b_2,b_3 \in \{-1,+1\}$ extracted from $T[x_i + j \alpha]$ tell us exactly in which interval $T[x_i + j \alpha]$ lies. Note that if we were to map those bits to $\{0,1\}$ (where $-1 \mapsto 0$ and $+1 \mapsto 1$), then the bit-string $b_1b_2b_3$ would correspond to the number associated with the interval and would thus be the correct corresponding input to the circuit.

We re-interpret the circuit $\lambda$ as working on bits $\{-1,+1\}$, where $-1$ corresponds to $0$. Clearly, we can implement this circuit in $C_j$ with our Boolean gates. As long as the inputs to every gate are perfect bits (i.e.,\ in $\{-1,+1\}$), the output of the gate will also be a perfect bit, and will correspond to the result of the operation computed by the gate. The inputs to the circuit will be exactly the bits obtained in the bit extraction phase for each $\widehat{x}_i$. Thus, it follows that if the bit extraction phase succeeds, then the simulation of $\lambda$ will always have a correct output. In other words, using \cref{clm:phase2}, we obtain:

\begin{claim}\label{clm:phase3}
Assume that the circuit-simulator $C_j$ is not corrupted and $\textup{const}_j = +1$. If $\textup{dist}(T[x_i + j \alpha], B) \geq 8g(\varepsilon)$ for all $i \in [N]$, then the simulation phase outputs $\overline{\lambda}(T[x + j \alpha])$.
\end{claim}

\paragraph{\textbf{Phase 4: output into the Feedback region.}}
For convenience, we will assume that the output of the Boolean circuit $\lambda$ is encoded in a particular way. It is easy to see that this is without loss of generality, since we can always modify $\lambda$ so that it follows this encoding. The output of $\lambda$ is an element in $\{\pm 1, \dots, \pm N\}$. The encoding we choose uses $2N$ bits $y_1^a,y_1^b, y_2^a,y_2^b, \dots, y_N^a,y_N^b$ to encode such an element. The element $+i$ is represented by $y_i^a = y_i^b = 1$ and $y_\ell^a = 1$, $y_\ell^b = 0$ for all $\ell \neq i$. The element $-i$ is represented by $y_i^a = y_i^b = 0$ and $y_\ell^a = 0$, $y_\ell^b = 1$ for all $\ell \neq i$.

Recall that in the simulation of $\lambda$ inside $C_j$, we actually use the bits $\{-1,+1\}$ instead of $\{0,1\}$. Thus, output $+i$ is represented by $y_i^a = y_i^b = +1$ and $y_\ell^a = +1$, $y_\ell^b = -1$ for all $\ell \neq i$. Whereas the element $-i$ is represented by $y_i^a = y_i^b = -1$ and $y_\ell^a = -1$, $y_\ell^b = +1$ for all $\ell \neq i$. For any $i \in [N]$ and any $z \in ([-1,1] \setminus B)^N$ define $\overline{\lambda}_i(z)$ to be:
\begin{itemize}
    \item $\overline{\lambda}_i(z) = +1$ if $\overline{\lambda}(z) = +i$
    \item $\overline{\lambda}_i(z) = -1$ if $\overline{\lambda}(z) = -i$
    \item $\overline{\lambda}_i(z) = 0$ otherwise.
\end{itemize}
Then, by \cref{clm:phase3} we obtain that $T[y_i^a + y_i^b] = \overline{\lambda}_i(T[x + j \alpha])$.

In this last phase, for each $i \in [N]$ we compute $T[y_i^a + y_i^b]$ and copy this value into the Feedback region, namely into interval $F_i(j)$ (recall that $C_j$ is the current circuit-simulator). For this we first use a $G_{+}$-gate and then a $G_{\textup{copy}}$-gate. Thus, we immediately obtain:

\begin{claim}\label{clm:phase4}
Assume that the circuit-simulator $C_j$ is not corrupted and $\textup{const}_j = +1$. If $\textup{dist}(T[x_i + j \alpha], B) \geq 8g(\varepsilon)$ for all $i \in [N]$, then $[C_j(x, \textup{const}_j)]_i := v(F_i(j)) = T[\overline{\lambda}_i(T[x + j \alpha]) \pm 2g(\varepsilon)]$ for all $i \in [N]$.
\end{claim}

\subsection{Proof of Correctness}\label{sec:ppacorrectness}

In this section we prove that the reduction works, i.e.,\ from any solution to the \textsc{Consensus-Halving} instance, we can obtain a solution to the original \textsc{high-D-Tucker} instance in polynomial time. In order to do this, we consider two cases and show that we can retrieve a solution in both cases. The first case corresponds to a ``well-behaved'' solution where there are no stray cuts. The second case corresponds to a solution with stray cuts. In both cases, we show that, if $x$ denotes the point encoded by the Coordinate-Encoding region in a solution of the Consensus-Halving instance, then the set $\{T[x + j \alpha] : j \in \{\pm 1,\dots, \pm p(N)\}\}$ must contain two points that have opposite labels in the original \textsc{high-D-Tucker} instance. Two such points can then easily be extracted in polynomial time by simply computing the labels of all $2p(N)$ points in that set. Note also that all points in the set are very close to each other by construction and so the two selected points will lie in adjacent hypercubes of the \textsc{high-D-Tucker} instance (as defined in Phase 3).\\

\noindent We set $p(N)=4N^2$ and pick $\varepsilon$ such that $16 g(\varepsilon) \leq \frac{1}{16p(N)}$, i.e.,\ $\varepsilon \leq \frac{1}{2^{14}N^2}$.

\begin{lemma}\label{lem:no-stray}
Let $S$ be any $\varepsilon$-approximate solution for the \textsc{Consensus-Halving} instance. If $S$ does not have any stray cuts, then it yields a solution to the \textsc{high-D-Tucker} instance in polynomial time.
\end{lemma}

\begin{proof}
Since there are no stray cuts, none of the circuit-simulators is corrupted. In particular, there is no cut strictly inside the Constant-Creation region. Thus, without loss of generality we can assume that the whole Constant-Creation region is labeled $+$. This means that all circuit-simulators read-in the constant $+1$, i.e.,\ $\textup{const}_j = +1$ for all $j$.

Since the circuit-simulators are not corrupted, the only way in which they can fail is if the bit extraction phase fails. Let $x \in [-1,1]^N$ be the point represented by the Coordinate-Encoding region in $S$. Let $z^{(j)} = T[x + j \alpha]$.

Consider any $i \in [N]$. First of all, note that if $z^{(j)}_i \in \{-1,+1\}$, then the bit extraction will not fail (for dimension $i$). So, we ignore any such points. For all $j \neq \ell$ such that $z^{(j)}_i, z^{(\ell)}_i \notin \{-1,+1\}$ it holds that
$$\alpha \leq |z^{(j)}_i - z^{(\ell)}_i| \leq p(N) \alpha.$$
Since $p(N) \alpha \leq 1/16$ and $\alpha = \frac{1}{16p(N)} \geq 16 g(\varepsilon)$, it follows that there exists at most one $j^*$ such that $z^{(j^*)}_i$ lies too close to $B = \{-3/4,-1/2,-1/4,0,1/4,1/2,3/4\}$. Namely, there exists $j^* \in [p(N)]$ such that for all $j \in [p(N)] \setminus \{j^*\}$, we have $\textup{dist}(z^{(j)}_i,B) \geq 8 g(\varepsilon)$. By \cref{clm:phase2} it follows that the bit extraction phase in dimension $i$ fails in at most one circuit-simulator.

We thus obtain that the bit extraction (in any dimension) fails in at most $n$ circuit-simulators. Let $X \subseteq [p(N)]$ be such that the bit extraction does not fail in $C_j$ for all $j \in X$. Then, we have shown that $|X| \geq p(N) - N \geq p(N) -2N$. By \cref{clm:phase3}, we get that for all $j \in X$, phase 3 of the circuit-simulator $C_j$ outputs $\overline{\lambda}(z^{(j)})$, i.e., the correct label for $z^{(j)}$ (which lies in $([-1,1] \setminus B)^N$).

Since $\|z^{(j)} - z^{(\ell)}\|_\infty \leq p(N) \alpha \leq 1/16$ for all $j, \ell \in [p(N)]$, all the points $z^{(j)}$ lie in (pairwise) adjacent hypercubes of $[-1,1]^N$ (as defined in phase 3). Thus, in order to find a solution to the \textsc{high-D-Tucker} instance, it suffices to find $j, \ell \in X$ such that $\overline{\lambda}(z^{(j)}) = - \overline{\lambda}(z^{(\ell)})$.

By \cref{clm:phase4}, we have that for every $j \in X$ the circuit-simulator $C_j$ outputs $$[C_j(x, \textup{const}_j)]_1, \dots, [C_j(x, \textup{const}_j)]_N$$ such that for all $i \in [N]$ we have $[C_j(x, \textup{const}_j)]_i \approx \overline{\lambda}_i(z^{(j)})$ (up to error $2g(\varepsilon)$). It holds that for every $j \in X$, there exists $i_j \in [N]$ such that 
$|\overline{\lambda}_{i_j}(z^{(j)})| = 1$. Thus, there exist $i \in [N]$ and $X_i \subseteq X$ such that 
\[
|X_i| \geq |X|/N, |\overline{\lambda}_i(z^{(j)})| = 1 \ \  \text{for all } \ \ j \in X_i \ \ \text{ and } \ \ \overline{\lambda}_i(z^{(j)})| = 0 \ \ \text{ for all } \ \ j \in X \setminus X_i.\] If there exist $j, \ell \in X_i$ such that $\overline{\lambda}_i(z^{(j)}) = - \overline{\lambda}_i(z^{(\ell)})$, then $z^{(j)}$ and $z^{(\ell)}$ have opposite labels and we are done. Thus, assume that $\overline{\lambda}_i(z^{(j)}) = +1$ for all $j \in X_i$ (the case with $-1$ instead works analogously). It follows that \[
[C_j(x, \textup{const}_j)]_i \geq 1 - 2g(\varepsilon) \ \ \text{ for all } j \in X_i.\] Since we also have $[C_j(x, \textup{const}_j)]_i \geq -2g(\varepsilon)$ for all $j \in X \setminus X_i$, we can write:
\begin{equation*}
\begin{split}
\sum_{j=1}^{p(N)} [C_j(x, \textup{const}_j)]_i &= \sum_{j \in X} [C_j(x, \textup{const}_j)]_i + \sum_{j \in [p(N)] \setminus X} [C_j(x, \textup{const}_j)]_i\\  
&\geq |X_i| - p(N)2g(\varepsilon) - |[p(N)] \setminus X|\\  
&\geq \frac{p(N)-2N}{N} - p(N)2g(\varepsilon) - 2N \geq 2N-2-2^7/2^{14} > p(N) \varepsilon
\end{split}
\end{equation*}
for $N$ sufficiently large, where we used $p(N) = 4N^2$ and $\varepsilon \leq \frac{1}{2^{14}N^2}$. However, recall that feedback agent $f_i$ ensures that $\sum_{j=1}^{p(N)} [C_j(x, \textup{const}_j)]_i = \sum_{j=1}^{p(N)} v(F_i(j)) \in [-p(N)\varepsilon, p(N)\varepsilon]$. So, we have obtained a contradiction.
\end{proof}

\begin{lemma}\label{lem:stray}
Let $S$ be any $\varepsilon$-approximate solution for the \textsc{Consensus-Halving} instance. If $S$ has at least one stray cut, then it yields a solution to the \textsc{high-D-Tucker} instance in polynomial time.
\end{lemma}

\begin{proof}
As explained earlier, stray cuts can affect the circuit-simulators in two ways. First of all, a stray cut can corrupt a circuit-simulator. Since there are at most $N$ stray cuts, at most $N$ circuit-simulators can be corrupted. The other way in which the circuit-simulators can be affected, is if there are stray cuts in the Constant-Creation region and $\textup{const}_j = -1$ for some $C_j$, and $\textup{const}_\ell = +1$ for some other $C_\ell$.

Let us now take a look at what happens if a circuit-simulator $C_j$ has $\textup{const}_j = -1$. Since all the gates in $C_j$ are equivariant (\cref{rem:equivariant-gates}), it follows that $C_j$ with inputs $x_1, \dots, x_N$ and $\textup{const}_j$ is also equivariant. This implies that 
\[
[C_j(x, -1)]_i = - [C_j(-x, +1)]_i \text{ for all } i \in [N].\]
Thus, the output of $C_j$ is the negation of a valid output that $C_j$ would have had if its inputs were $-x$ and $\textup{const}_j = +1$ instead. We use the term ``a valid output'' instead of ``the output'', because $C_j$ can have multiple valid outputs for a single input (because every gate is allowed a small error). So, we are saying that the output of $C_j$ on inputs $x,-1$ is the negation of a valid output on inputs $-x,+1$. Thus, we immediately obtain that the analogous statements for Claims \ref{clm:phase1},\ref{clm:phase2},\ref{clm:phase3} and \ref{clm:phase4} also hold in this case. In other words, we get that if circuit-simulator $C_j$ is not corrupted, and if $\textup{dist}(T[x_i + j \alpha], B) \geq 8g(\varepsilon)$, then 
\[
[C_j(x, +1)]_i = \overline{\lambda}_i(T[x + j \alpha]) \pm 2g(\varepsilon) \ \ \text{ and } \ \  [C_j(x, -1)]_i = -\overline{\lambda}_i(T[-x + j \alpha]) \pm 2g(\varepsilon).
\]
Consider any $i \in [N]$. Using the same argument as in the proof of \cref{lem:no-stray}, it is easy to show that at most one of the points 
\[
T[x_i + \alpha], \dots, T[x_i + p(N) \alpha] \ \  \text{ and } \ \ T[x_i - \alpha], \dots, T[x_i - p(N) \alpha]\] lies within distance at most $8g(\varepsilon)$ of $B$. It follows that at most one of the points
\[
T[x_i + \alpha], \dots, T[x_i + p(N) \alpha] \ \ \text{ and } \ \  T[-x_i + \alpha], \dots, T[-x_i + p(N) \alpha]
\]lies within distance at most $8g(\varepsilon)$ of $B$. Thus, we again have that the bit extraction phase fails in at most $N$ circuit-simulators.

Let $X \subseteq [p(N)]$ be such that for all $j \in X$, $C_j$ is not corrupted and the bit extraction does not fail in $C_j$. Then we have shown that $|X| \geq p(N) - 2N$. The same argument as in the proof of \cref{lem:no-stray} yields that there must exist $i \in [N]$ and $j, \ell \in X$ such that 
\[
[C_j(x, \textup{const}_j)]_{i} = +1 \pm 2g(\varepsilon) \ \  \text{ and } \ \ [C_\ell(x, \textup{const}_\ell)]_{i} = -1 \pm 2g(\varepsilon).\]
If $\textup{const}_j = \textup{const}_\ell = +1$, then it follows that \[
\overline{\lambda}_i(T[x + j \alpha]) = +1\ \ \text{ and } \ \  \overline{\lambda}_i(T[x + \ell \alpha]) = -1.\]
Thus, we have obtained two points in adjacent hypercubes that have opposite labels. If $\textup{const}_j = \textup{const}_\ell = -1$, then it follows that \[
\overline{\lambda}_i(T[-x + j \alpha]) = -1 \ \ \text{ and } \ \ \overline{\lambda}_i(T[-x + \ell \alpha]) = +1.\] Thus, we again have obtained two points in adjacent hypercubes that have opposite labels.

Let us now see what happens in the case where $\textup{const}_j = - \textup{const}_\ell$. Without loss of generality assume that $\textup{const}_j = +1$ and $\textup{const}_\ell = -1$. Then, we obtain that $\overline{\lambda}(T[x+j \alpha]) = \overline{\lambda}(T[-x+\ell \alpha])$. Since there is at least one stray cut in the solution $S$, by \cref{obs:stray} we know that the point $x \in [-1,1]^N$ encoded by the Coordinate-Encoding region must lie on the boundary of $[-1,1]^N$. Thus, since $p(N) \alpha \leq 1/16$, $u=T[x+j \alpha]$ and $v=T[-x+\ell \alpha]$ must lie in hypercubes on the boundary. Since $2 p(N) \alpha \leq 1/8$, we know that $u$ and $-v$ lie in adjacent hypercubes. However, by the boundary conditions of the \textsc{high-D-Tucker} instance $\lambda$, we have $\overline{\lambda}(u) = \overline{\lambda}(v) = -\overline{\lambda}(-v)$. Thus, $u$ and $-v$ yield adjacent hypercubes with opposite labels, and thus a solution to the \textsc{high-D-Tucker} instance. Note that $u$ and $-v$ cannot lie in the same hypercube, because of the boundary conditions of $\lambda$.
\end{proof}

\section{Algorithms for single-block valuations}

In the previous section, we proved that even when we have $2$-block uniform valuations, the problem remains PPA-hard (even when we are allowed to use $n + n^{1-\delta}$ cuts, for some constant $\delta>0$). In this section, we will consider the natural case that is not covered by our hardness, that of single-block valuations. Our main results of the section are (a) an algorithm for solving the problem to any precision $\varepsilon$ (where $\varepsilon$ appears polynomially in the running time), which is parameterized by the maximum number of intersection between the blocks of different agents and (b) a polynomial-time algorithm for $1/2$-\ch. The latter algorithm generalizes to the case of $d$-block uniform valuations (in fact, even to the case of piecewise constant valuations with $d$ blocks) if one is allowed to use $d\cdot n$ cuts instead of $n$. Towards the end of the section, we also present a simple idea based on linear programming, which allows us to solve the \ch\ problem in polynomial time, when we are allowed to use $2n-\ell$ cuts, for any $\ell$ which is constant.

\subsection{A parameterized algorithm for \texorpdfstring{$\varepsilon$}{ε}-\ch}\label{sec:dpalgo}

We start with the algorithm for 
solving $\varepsilon$-\ch\ that is parameterized by
the \textit{maximum intersection} between the 
probability measures. In particular, if $d$ is the
maximum number of measures with positive density at
any point $x \in [0, 1]$ and the maximum value of the
densities is at most $M$ then we provide a dynamic
programming algorithm that computes an
$\varepsilon$-approximate solution in time 
$O\left(\left(\frac{M}{\varepsilon}\right)^d \cdot \mathrm{poly}(M/\varepsilon, n, d) \right)$.
We start with the formal definition of the maximum 
intersection quantity. In this section we denote by 
$f_i$ the probability density function of the 
probability measure $\mu_i$.

\begin{definition}[\textsc{Maximum Intersection \& Maximum Value}]
    We say that a single-block instance of 
  $\varepsilon$-\ch\ has 
  \textit{maximum intersection} $d$ if for every
  $x \in [0, 1]$ it holds that the set
  $\mathcal{R}(x) = \{i \in [n] \mid f_i(x) > 0\}$,
  has cardinality 
  $\left|\mathcal{R}(x)\right| \le d$. We also
  say that the instance has 
  \textit{maximum value} $M$ if for every 
  $i \in [n]$ and every $x \in [0, 1]$ it holds
  that $f_i(x) \le M$. This is equivalent to
  $b_i - a_i \ge 1/M$.
\end{definition}

\noindent For the rest of the section, we often refer
to the following quantity.

\begin{definition}[\textsc{Value of Balance}]
    For any $\varepsilon$-\ch\ instance 
  $\mathcal{C} = \{\mu_1, \dots, \mu_n\}$, any
  vector of cuts $\vec{s}$ and any 
  $z \in [0, 1]$ let 
  $b_i(\vec{s}; z)$ be the mass with respect to 
  $\mu_i$ of the part of the interval $[z, 1]$ labeled
  ``$+$'' minus the part of the interval labeled
  ``$-$'', when split with the cuts $\vec{s}$.
  Formally, $b_i(\vec{s}; z) = \mu_i([z,1]^{+}) - \mu_i([z,1]^{-})$,
  given the set of cuts $\vec{s}$.
\end{definition}
  
  The first step of the algorithm is to 
discretize the interval $[0, 1]$ into
intervals of length $\varepsilon/(2 M)$. Hence, we 
split the interval $[0, 1]$ into 
$m = 2 M/\varepsilon$ equal subintervals of the 
form $p_{\ell} = [(\ell - 1)/m, \ell/m]$, where
$\ell \in [m]$. The following claim shows the 
sufficiency of our discretization and follows 
very easily from the above definitions.

\begin{claim} \label{clm:sufficientDiscretization}
    Let 
  $Q_m = \left\{\frac{\ell - 1}{m} \mid \ell \in [m] \right\}$,
  where $m = M/\varepsilon'$ and let 
  $\mathcal{C} = \{\mu_1, \dots, \mu_n\}$ be a
  single-block $\varepsilon$-\ch\ instance with 
  maximum value $M$. Then we define the
  \emph{rounded} single-block instance
  $\mathcal{C}' = \{\mu'_1, \dots, \mu'_n\}$ 
  where $a'_i$ and $b'_i$ are equal to the
  number of $Q_m$ that is closer to $a_i$ and 
  $b_i$ respectively. Then every 
  $\varepsilon$-\ch\ solution of
  $\mathcal{C}'$ is an $(\varepsilon + \varepsilon')$-\ch\ solution of
  $\mathcal{C}$.
  Additionally, there exists a solution to the
  $\varepsilon'$-\ch\ problem where
  the positions of the cuts lie in the set
  $Q_m$.
\end{claim}

  Because of Claim 
\ref{clm:sufficientDiscretization} we will 
assume for the rest of this section that we are
working with $\mathcal{C}'$ and we will focus
on finding an $(\varepsilon' = \varepsilon/2)$-\ch\
solution with cuts in $Q_m$, where 
$m = 2 M /\varepsilon$. This will give us an 
$\varepsilon$-approximate solution for the 
single-block instance $\mathcal{C}$. We also
need the following definitions:
\begin{itemize}
  \item[$\triangleright$] for any $z \in [0, 1]$
  and any instance 
  $\mathcal{C} = \{\mu_1, \dots, \mu_n\}$ we 
  define the set of measures that have positive
  mass in $[z, 1]$ as follows:
  $\mathcal{U}(z; \mathcal{C}) = \{i \in [n] \mid \mu_i \in \mathcal{C} \wedge \mu_i([z, 1]) > 0\}$
  where we might drop $\mathcal{C}$ when it is 
  clear from the context, see also Figure
  \ref{fig:dynamic:sets},
  \item[$\triangleright$] for any $z \in [0, 1]$
  and any instance 
  $\mathcal{C} = \{\mu_1, \dots, \mu_n\}$ we 
  define the set of measures that have positive
  density at $z$ as follows:
  $\mathcal{R}(z; \mathcal{C}) = \{i \in [n] \mid \mu_i \in \mathcal{C} \wedge f_i(z) > 0\}$
  where we might drop $\mathcal{C}$ when it is 
  clear from the context, see also Figure
  \ref{fig:dynamic:sets}.
\end{itemize}

\begin{figure}
  \centering
  \includegraphics[scale=0.45]{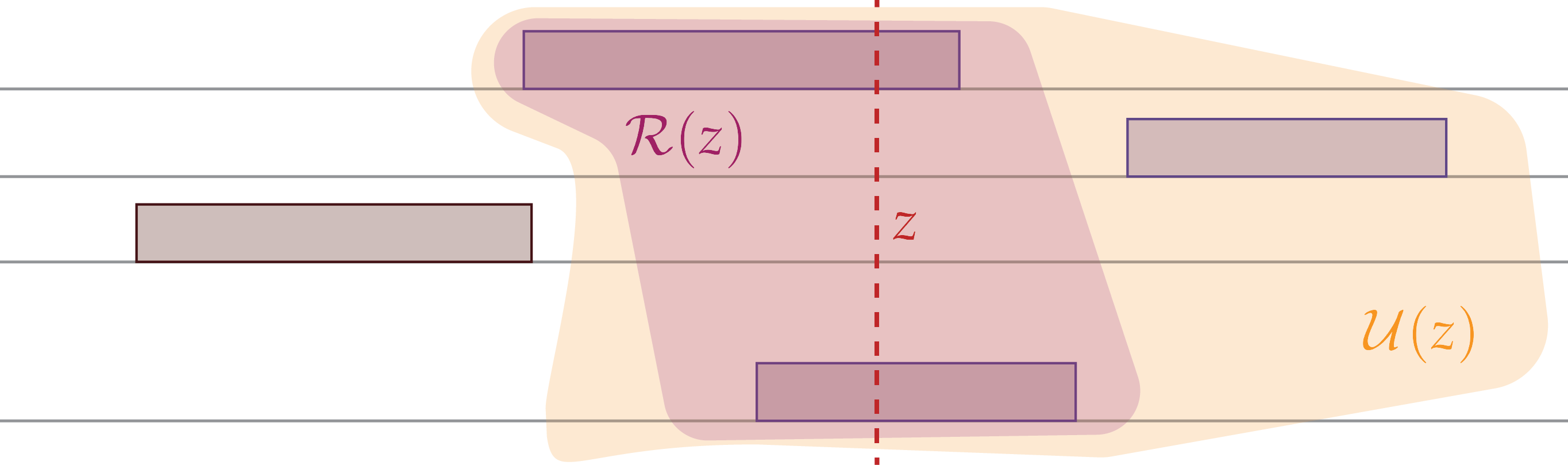}
   \caption{The definition of the sets $\mathcal{U}(z)$  and $\mathcal{R}(z)$.}
  \label{fig:dynamic:sets}
\end{figure}

  We also define 
$\bar{Q}_m = \{z \mid z \in Q_m \vee -z \in Q_m\}$.
Now we are ready to define our main recursive
relation for solving the instance 
$\mathcal{C}'$. For this we define the function
$\alpha(q_1, \dots, q_d, z, t)$ where 
$q_j \in \bar{Q}_m$, $z \in Q_m$, and 
$t \in [n]$ and 
$\alpha : \bar{Q}_m^d \times Q_m \times [n] \to \{0, 1\}$. 
The intuitive explanation of the value of
$\alpha$ is the following:
\smallskip

\begin{quote}
  \emph{
    $\alpha(q_1, \dots, q_d, z, t)$ denotes 
  whether it is possible to find $t$ cuts 
  $\vec{s}$ to split the interval $[z, 1]$ such
  that: (1) if $i$ is the $j$th element of~
  $\mathcal{R}(z)$ it holds that 
  $|b_i(\vec{s}; z) - q_j| \le \varepsilon$, (2)
  for every 
  $i \in \mathcal{U}(z) \setminus \mathcal{R}(z)$ 
  it holds that $|b_i(\vec{s}; 0)| = |b_i(\vec{s}; z)| \le \varepsilon$. 
  }
\end{quote}
\smallskip

  The value of $\alpha(q_1, \dots, q_d, z, t)$
can be recursively computed via the following
procedure.

\begin{enumerate}
    \item for every $r \in Q_m$, with $r > z$ do
    \begin{itemize}
      \item[a.] if there exists 
      $i \in \mathcal{R}(z)$ and $i$ is the
      $j$th element of $\mathcal{R}(z)$ but 
      $i \not\in \mathcal{R}(r)$ and by adding
      the cut $r$ to the set of cuts then the 
      condition 
      $|b_i(r; z) - q_j| > \varepsilon$ holds 
      then \textbf{continue},
      \item[b.] if there exists 
      $i \in \mathcal{U}(z) \setminus \mathcal{R}(z)$
      but 
      $i \not\in \mathcal{U}(r) \setminus \mathcal{R}(r)$ 
      then \textbf{continue}
      \item[c.] else for every 
      $i \in \mathcal{R}(r)$, where $i$ is the
      $j$th element of $\mathcal{R}(r)$
      \begin{itemize}
        \item[$\cdot$] if $i \not\in \mathcal{R}(z)$ then we set 
        $q'_j = (-1)^t \mu_i([z, r])$
        \item[$\cdot$] if $i$ is the $\ell$th element of $\mathcal{R}(z)$ 
        then we set 
        $q'_j = (-1)^t \mu_i([z, r]) + q_{\ell}$
        \item[$\cdot$] call the function $\alpha(q'_1, \dots, q'_d, r, t - 1)$
      \end{itemize}
    \end{itemize}
    \item return \textbf{true} if at least one of the recursive calls is successful, and \textbf{false} otherwise.
\end{enumerate}

  This procedure evaluates the binary 
function $\alpha$, but our goal is to solve 
the search problem that finds a set of cuts
that form a solution to 
$\varepsilon'$-\ch. This can be 
easily done by storing one possible solution 
whenever $\alpha = \textbf{true}$. We call 
this possible solution 
$\beta(q_1, \dots, q_d, z, t)$. More formally,
the algorithm is shown in \cref{alg:mainBreimanPopulation}.
    
\begin{algorithm}[h]
\caption{Recursive Computation of $\alpha$, $\beta$}
  \KwIn{leftover balances $q_1, \dots, q_d$, position of last cut $z$, leftover number of cuts $t$, required accuracy $\varepsilon$.}

  \KwOut{value of $(\alpha(q_1, \dots, q_d, z, t), \beta(q_1, \dots, q_d, z, t))$ as described above.}

  \If{$t = 0$}{
    \If{$q_j = 0~\forall j \in [d]$}{
      \textbf{return} $(\textbf{true}, \{\})$
    }
    \textbf{return} $(\textbf{false}, \{\})$
  }
  
  \For{every $r \in Q_m$ with $r > z$}{
    \For{$j \in [d]$}{
      Let $i$ be the $j$th element of $\mathcal{R}(z)$
      
      \If{$i \not\in \mathcal{R}(r)$ and $|b_i(r; z) - q_j| > \varepsilon$}{
        \textbf{continue} to next value of $r$
      }
    }
    \For{$i \in \mathcal{U}(z) \setminus \mathcal{R}(z)$}{
      \If{$i \not\in \mathcal{U}(r) \setminus \mathcal{R}(r)$}{
        \textbf{return} $(\textbf{false}, \{\})$
      }
    }
    \For{$j \in [d]$}{
      Let $i$ the $j$th element of $\mathcal{R}(r)$
      
      \If{$i \not\in \mathcal{R}(z)$}{
        $q'_j \leftarrow (-1)^t \mu_i([z, r])$
      }
      \If{$i$ is the $\ell$th element of $\mathcal{R}(z)$}{
        $q'_j \leftarrow (-1)^t \mu_i([z, r]) + q_{\ell}$
      }
    }
    $\mathrm{res} \leftarrow (\alpha(q'_1, \dots, q'_d, r, t - 1), \beta(q'_1, \dots, q'_d, r, t - 1))$
    
    \If{$\mathrm{res} = \textbf{\emph{true}}$}{
      \textbf{return} $(\textbf{true}, \beta(q'_1, \dots, q'_d, r, t - 1) \cup\{r\})$
    }
  }
  
  \textbf{return} $(\textbf{false}, \{\})$
  \label{alg:mainBreimanPopulation}
\end{algorithm}

  From the above recursive algorithm we see
that the evaluation order for the dynamic 
programming algorithm that computes 
$\alpha(q_1, \dots, q_d, z, t)$ starts from
$t = 0$ to $t = n$, $z = 1$ to $z = 0$ and 
$|q_j| = 1$ to $|q_j| = 0$.

\begin{theorem} \label{thm:dynamicProgramming}
    There exists a dynamic programming 
  algorithm that for any single-block 
  instance $\mathcal{C}$ of $\varepsilon$-\ch\ 
  with maximum value $M$ and maximum intersection 
  $d$, computes a set of cuts that define a
  solution. The running time of the algorithm
  is 
  $O\left( \left(\frac{2 M}{\varepsilon}\right)^{d + 2} n (n + d) \right)$.
\end{theorem}

\subsection{A polynomial-time algorithm for \texorpdfstring{$1/2$}{1/2}-\ch}

In this subsection, we present a polynomial-time algorithm for $1/2$-\ch, for the case of single-block valuations. We state the main theorem, which will be proven in the subsection.

\begin{theorem}
There exists a polynomial-time algorithm for $1/2$-\ch, when agents have single-block valuations.
\end{theorem}

\paragraph{\textbf{High-level description:}} The high-level idea of the algorithm is the following greedy strategy. We will consider the agents in order of non-decreasing height of their valuation blocks. Since each agent has a single block of value, this means that for two agents $i$ and $j$ with $i < j$ that have their value block in intervals $I_i$ and $I_j$, agent $i$ will be considered before agent $j$. For each agent, we will attempt to ``reserve'' a large enough sub-interval of her value block (of total value at least $1/2$ for the agent) and split it in half (to two parts of equal value at least $1/4$ to the agent) using a cut, assigning each half to $+$ and $-$ respectively. At that point, this split will ensure that $|\mu_i(I_i^{+}) - \mu_i(I_i^{-})| \leq 1/2$, regardless of the labeling of the remaining part of the agent's value block $I_i$.
A reserved sub-interval, or \emph{reserved region} (RR) will never be intersected by any subsequent cut. This ensures that the guarantee $|\mu_i(I_i^{+}) - \mu_i(I_i^{-})| \leq 1/2$ for every agent $i$ previously considered will continue to hold. 

Reserving a large enough region for the first considered agent is straightforward. For any subsequent agent $i$, part of her value block interval $I_i$ might already be ``covered'' by regions that were reserved in previous steps. Before inserting any new cut, we will \emph{expand} some of the RRs which are contained in $I_i$ (those that exhibit a different label on each of their endpoints), until we either ensure that the agent is approximately satisfied with the imbalance of labels that she sees, or these RRs cannot be expanded any longer. In the latter case, we will place the cut corresponding to agent $i$ in $I_i$, on the midpoint of the ``virtual'' interval $U_i \subseteq I_i$, consisting of all the intervals of $I_i$ not covered by RRs, ``glued'' together. Then, we will create a new RR, which will potentially also contain some of the RRs already present in $I_i$, which will be such that (a) the total value covered by RRs for agent $i$ is $1/2$ and (b) the agent is approximately satisfied (up to a $1/2$) with the value she has for RRs in $I_i$.

\subsubsection{The algorithm} 
First, we will use the following terminology: A \emph{reserved region} (RR) $R$ is a designated interval in which no further cuts are allowed to be placed, other than those that already lie in the region. We will define the \emph{parity} of $R$ to be \emph{odd} if the left-hand side of the leftmost cut intersecting $R$ and the right-hand side of the rightmost cut intersecting $R$ have different labels; otherwise, we will say that the parity of $R$ is \emph{even}. 

We will let $I_i$ denote the (single) interval in which agent $i$ has positive value, and we will let $R_i = \bigcup_{R_j \text{ is an RR}} (R_j \cap I_i)$ be the part of $I_i$ that is ``covered'' by RRs. We will also let $U_i = I_i \setminus R_i$ denote the \emph{unreserved} part of $I_i$. Finally, we will say that an RR $R$ is an \emph{internal RR} of $I_i$, if it is entirely contained in $I_i$ (including the case where the endpoint of the RR is the endpoint of the interval); otherwise, we will say that $R$ is a \emph{boundary RR}. Note that since RRs can be contained in other RRs, it is possible that an internal RR is contained in a boundary RR (see \cref{fig:rrs}). An overview of the algorithm is given in \cref{alg:halfch}.\\

    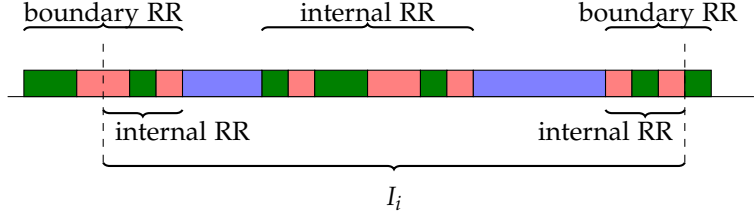
\begin{figure}
\centering
\begin{tikzpicture}[scale=1,transform shape]
	\node (a_1) at (-20pt,0pt) {}; 
	\node (a_2) at (270pt, 0pt) {};
	\draw (a_1)--(a_2);

	\draw[fill=black!50!green] (-10pt,0pt) rectangle (10pt,10pt);
	\draw[fill=white!50!red] (10pt,0pt) rectangle (30pt,10pt);
    \draw[fill=black!50!green] (30pt,0pt) rectangle (40pt,10pt);
    \draw[fill=white!50!red] (40pt,0pt) rectangle (50pt,10pt);
    \draw[fill=black!50!green] (80pt,0pt) rectangle (90pt,10pt);
    \draw[fill=white!50!red] (90pt,0pt) rectangle (100pt,10pt);
     \draw[fill=black!50!green] (100pt,0pt) rectangle (120pt,10pt);
    \draw[fill=white!50!red] (120pt,0pt) rectangle (140pt,10pt);
    \draw[fill=black!50!green] (140pt,0pt) rectangle (150pt,10pt);
    \draw[fill=white!50!red] (150pt,0pt) rectangle (160pt,10pt);

	\draw[fill=white!50!red] (210pt,0pt) rectangle (220pt,10pt);
	\draw[fill=black!50!green] (220pt,0pt) rectangle (230pt,10pt);
	\draw[fill=white!50!red] (230pt,0pt) rectangle (240pt,10pt);
	\draw[fill=black!50!green] (240pt,0pt) rectangle (250pt,10pt);
	
	\draw[fill=white!50!blue] (160pt,0pt) rectangle (210pt,10pt);
	\draw[fill=white!50!blue] (50pt,0pt) rectangle (80pt,10pt);

\draw [
    thick,
    decoration={
        brace,
        mirror,
        raise=5pt
    },
    decorate
] (20pt,-20pt) -- (240pt,-20pt)
node [pos=0.5,anchor=north,yshift=-10pt] {\small{$I_i$}};

\draw [
    thick,
    decoration={
        brace,
        raise=-15pt
    },
    decorate
] (80pt,40pt) -- (160pt,40pt)
node [pos=0.5,anchor=north,yshift=0pt] {\small{internal RR}};

\draw [
    thick,
    decoration={
        brace,
        mirror,
        raise=15pt
    },
    decorate
] (20pt,10pt) -- (50pt,10pt)
node [pos=0.5,anchor=north,xshift=10pt, yshift=-15pt] {\small{internal RR}};

\draw [
    thick,
    decoration={
        brace,
        mirror,
        raise=15pt
    },
    decorate
] (210pt,10pt) -- (240pt,10pt)
node [pos=0.5,anchor=north,xshift=-10pt, yshift=-15pt] {\small{internal RR}};

\draw [
    thick,
    decoration={
        brace,
        raise=-15pt
    },
    decorate
] (-10pt,40pt) -- (50pt,40pt)
node [pos=0.5,anchor=north,yshift=0pt] {\small{boundary RR}};

\draw [
    thick,
    decoration={
        brace,
        raise=-15pt
    },
    decorate
] (210pt,40pt) -- (250pt,40pt)
node [pos=0.5,anchor=north,yshift=0pt] {\small{boundary RR}};

\draw[dashed] (20pt,20pt) -- (20pt, -25pt);
\draw[dashed] (240pt,20pt) -- (240pt, -25pt);

\end{tikzpicture} \caption{Internal and external reserved areas. The areas colored green have been assigned to one of the labels (e.g., ``+'') and the areas colored red had been assigned to the other (e.g., ``-''). The areas of the interval $I_i$ which are not yet assigned to any label are colored blue. The parity of the RRs can also be seen: the boundary RRs and the internal RR in the middle of the interval have odd parity, whereas the other two internal RRs (which are sub-regions of the boundary RRs) have even parity.}
\label{fig:rrs}
\end{figure}

\noindent The algorithm considers the agents in non-increasing order of the heights of their value blocks. We will say that an agent $i$ is \emph{$1/2$-satisfied at step $j$}, if, after considering agent $j$, it holds that $|\mu_i(I^{+}) - \mu_i(I^{-}) \leq 1/2$, i.e., at least half of the value of the agent is balanced between $+$ and $-$. Let step $i$ be the step when agent $i$ is considered. At step $i$ the algorithm does the following.\medskip
\begin{enumerate}
    \item Expand internal RRs of odd parity, until agent $i$ is $1/2$-satisfied, or until any internal RRs of odd parity can no longer be extended. See \cref{sec:oddparity} below.\medskip
    \item Place the cut in the midpoint of $\mathcal{U}_i$ (where $\mathcal{U}_i$ is obtained by considering $U_i$ as a continuous interval, and disregarding $R_i$). See \cref{sec:placecut} below.\medskip
    \item Create a new RR, by considering only $U_i$, possibly merging the intervals of the new RR with some RRs in $R_i$. See \cref{sec:virtualrr} below.
\end{enumerate}

\subsubsection{Expanding internal Reserved Regions of odd parity}\label{sec:oddparity} Consider step $i$, when we are considering agent $i$, and her corresponding value block interval $I_i$. We will consider internal RRs of odd parity in $I_i$. For any such RR, we extend its endpoints at the same rate towards the left and the right, until we either:
    \begin{enumerate}[label=(\alph*)]
        \item reach a point where the agent is $1/2$-satisfied or \label{item1}
        \item reach the endpoint of one or two other RRs, or, \label{item2}
        \item reach one of the endpoints of $I_i$. \label{item3}
    \end{enumerate}

  \noindent If \cref{item1} above is satisfied, then we do not consider any other RR and continue with the next agent. Otherwise, we continue with the next internal RR of odd parity, if any. Note that after this part of the algorithm, the cut corresponding to agent $i$ has not been placed yet. See \cref{fig:oddRRexpand}.
   
    \begin{figure}
    \centering
        \begin{tikzpicture}[scale=1,transform shape]
	\node (a_1) at (-20pt,0pt) {}; 
	\node (a_2) at (240pt, 0pt) {};
	\draw (a_1)--(a_2);

	\draw[fill=black!50!green] (-10pt,0pt) rectangle (10pt,10pt);
	\draw[fill=white!50!red] (10pt,0pt) rectangle (30pt,10pt);
    \draw[fill=black!50!green] (30pt,0pt) rectangle (40pt,10pt);
    \draw[fill=white!50!red] (40pt,0pt) rectangle (50pt,10pt);
    \draw[fill=black!50!green] (80pt,0pt) rectangle (90pt,10pt);
    \draw[fill=white!50!red] (90pt,0pt) rectangle (100pt,10pt);
     \draw[fill=black!50!green] (100pt,0pt) rectangle (120pt,10pt);
    \draw[fill=white!50!red] (120pt,0pt) rectangle (140pt,10pt);
    \draw[fill=black!50!green] (140pt,0pt) rectangle (150pt,10pt);
    \draw[fill=white!50!red] (150pt,0pt) rectangle (160pt,10pt);

	\draw[fill=white!50!blue] (160pt,0pt) rectangle (240pt,10pt);
	\draw[fill=white!50!blue] (50pt,0pt) rectangle (80pt,10pt);

	\node (a_1) at (-20pt,-30pt) {}; 
	\node (a_2) at (240pt, -30pt) {};
	\draw (a_1)--(a_2);

    \draw[fill=black!50!green] (-10pt,-30pt) rectangle (10pt,-20pt);
	\draw[fill=white!50!red] (10pt,-30pt) rectangle (30pt,-20pt);
    \draw[fill=black!50!green] (30pt,-30pt) rectangle (40pt,-20pt);
    \draw[fill=white!50!red] (40pt,-30pt) rectangle (50pt,-20pt);
    \draw[fill=black!50!green] (80pt,-30pt) rectangle (90pt,-20pt);
    \draw[fill=white!50!red] (90pt,-30pt) rectangle (100pt,-20pt);
     \draw[fill=black!50!green] (100pt,-30pt) rectangle (120pt,-20pt);
    \draw[fill=white!50!red] (120pt,-30pt) rectangle (140pt,-20pt);
    \draw[fill=black!50!green] (140pt,-30pt) rectangle (150pt,-20pt);
    \draw[fill=white!50!red] (150pt,-30pt) rectangle (160pt,-20pt);

	\draw[fill=white!50!blue] (190pt,-20pt) rectangle (240pt,-30pt);
	
	\draw[fill=white!50!red] (160pt,-20pt) rectangle (190pt,-30pt);
\draw[fill=black!50!green] (50pt,-20pt) rectangle (80pt,-30pt);
	
	\node (a_3) at (255pt,0pt) {$\ldots$};
	\node (a_4) at (255pt,-30pt) {$\ldots$};

\draw [
    thick,
    decoration={
        brace,
        mirror,
        raise=5pt
    },
    decorate
] (20pt,-50pt) -- (270pt,-50pt)
node [pos=0.5,anchor=north,yshift=-10pt] {\small{$I_i$}};

\draw [
    thick,
    decoration={
        brace,
        raise=-15pt
    },
    decorate
] (80pt,30pt) -- (160pt,30pt)
node [pos=0.5,anchor=north,yshift=0pt] {\small{$R$}};

\draw [
    thick,
    decoration={
        brace,
        raise=-15pt
    },
    decorate
] (50pt,0pt) -- (190pt,0pt)
node [pos=0.5,anchor=north,yshift=0pt] {\small{$R'$}};

\draw[dashed] (20pt,40pt) -- (20pt, -55pt);
\draw[dashed] (270pt,40pt) -- (270pt, -55pt);

\end{tikzpicture}           \caption{An expansion of an internal reserved region of odd parity. The areas colored green have been assigned to one of the labels (e.g., ``+'') and the areas colored red had been assigned to the other (e.g., ``-''). The areas of the interval $I_i$ which are not yet assigned to any label are colored blue. The reserved region $R$ is expanded symmetrically on both sides, until it meets the left endpoint of the boundary region of $I_i$ on the left, in which case the expansion stops (\cref{item2} in \cref{sec:oddparity}). This results in the creation of new RRs on the left (a boundary RR and an internal RR contained in the boundary RR) that contain both the original RR $R$ and its expansion $R'$. Crucially, the area of $I_i$ that is covered by RRs has increased, while maintaining the same balance between the two labels.}
        \label{fig:oddRRexpand}
    \end{figure}
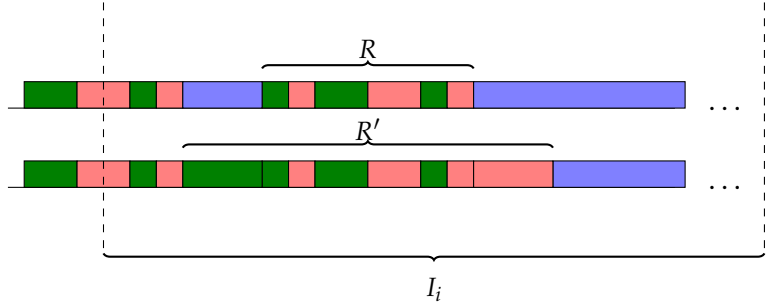

\subsubsection{Placing the cut}\label{sec:placecut} After we consider all internal RRs of odd parity, and if \cref{item1} above has not been satisfied yet, we will use the agent's corresponding cut to balance out the remaining value block $I_i$, to ensure that the agent becomes $1/2$-satisfied. To do that, we consider agent $i$'s valuation in $I_i$ on $U_i$, and we place the cut on the midpoint $y$ of $U_i$ - one can envision this operation as considering $U_i$ as a continuous interval (merging any sub-intervals separated by RRs in $R_i$) and splitting this interval in half via the placement of the cut. See \cref{fig:placecut}.

\subsubsection{Creating the new Reserved Region} \label{sec:virtualrr} After the cut has been placed, we need to reserve a corresponding region for agent $i$, to ensure that she is $1/2$-satisfied from the updated union of reserved regions $R_i$ in $I_i$. To do that, we reserve an equal portion of $U_i$ to the left and to the right of the position $y$ of the cut, until the agent becomes $1/2$-satisfied. We argue in the proof of \cref{lem:satisfied} that this is always possible. Intuitively, one can visualize that we extend the region from $y$ to the left and to the right at the same rate, ``jumping over'' RRs, until we have reserved enough area to $1/2$-satisfy the agent. The effect of this process is a new RR, possibly containing previously reserved RRs and the newly reserved regions. See \cref{fig:placecut}.

 \begin{figure}
   \centering
        \begin{tikzpicture}[scale=1,transform shape]
	\node (a_1) at (-20pt,0pt) {}; 
	\node (a_2) at (20pt, 0pt) {};
	\draw (a_1)--(a_2);

	\draw[fill=white!50!red] (20pt,0pt) rectangle (30pt,10pt);
    \draw[fill=black!50!green] (30pt,0pt) rectangle (40pt,10pt);
    \draw[fill=white!50!red] (40pt,0pt) rectangle (50pt,10pt);
    \draw[fill=black!50!green] (80pt,0pt) rectangle (90pt,10pt);
    \draw[fill=white!50!red] (90pt,0pt) rectangle (100pt,10pt);
    \draw[fill=black!50!green] (100pt,0pt) rectangle (110pt,10pt);
    \draw[fill=white!50!red] (110pt,0pt) rectangle (120pt,10pt);
    \draw[fill=black!50!green] (120pt,0pt) rectangle (130pt,10pt);

	\draw[fill=white!50!blue] (130pt,0pt) rectangle (180pt,10pt);
	\draw[fill=white!50!red] (180pt,0pt) rectangle (240pt,10pt);
	\draw[fill=white!50!blue] (50pt,0pt) rectangle (80pt,10pt);

	\node (a_1) at (-20pt,-50pt) {}; 
	\node (a_2) at (240pt, -50pt) {};
	\draw (a_1)--(a_2);

	\draw[fill=white!50!red] (20pt,-50pt) rectangle (30pt,-40pt);
    \draw[fill=black!50!green] (30pt,-50pt) rectangle (40pt,-40pt);
    \draw[fill=white!50!red] (40pt,-50pt) rectangle (50pt,-40pt);
    \draw[fill=black!50!green] (80pt,-50pt) rectangle (90pt,-40pt);
    \draw[fill=white!50!red] (90pt,-50pt) rectangle (100pt,-40pt);
    \draw[fill=black!50!green] (100pt,-50pt) rectangle (110pt,-40pt);
    \draw[fill=white!50!red] (110pt,-50pt) rectangle (120pt,-40pt);
    \draw[fill=black!50!green] (120pt,-50pt) rectangle (130pt,-40pt);

	\node (a_3) at (0pt,0pt) {$\ldots$};
	\node (a_4) at (0pt,-30pt) {$\ldots$};

	\draw[fill=white!50!blue] (140pt,-50pt) rectangle (180pt,-40pt);
	\draw[fill=white!50!red] (180pt,-50pt) rectangle (240pt,-40pt);
	\draw[fill=white!50!blue] (50pt,-50pt) rectangle (80pt,-40pt);

	\draw[fill=black!20!green] (130pt,-50pt) rectangle (140pt,-40pt);
	\draw[fill=black!20!green] (70pt,-50pt) rectangle (80pt,-40pt);
	
	\draw[fill=white!70!red] (140pt,-50pt) rectangle (160pt,-40pt);
	
	\node (a_3) at (255pt,0pt) {$\ldots$};
	\node (a_4) at (255pt,-30pt) {$\ldots$};

\draw [
    thick,
    decoration={
        brace,
        mirror,
        raise=5pt
    },
    decorate
] (20pt,-80pt) -- (240pt,-80pt)
node [pos=0.5,anchor=north,yshift=-10pt] {\small{$I_i$}};

\draw [
    thick,
    decoration={
        brace,
        raise=-15pt
    },
    decorate
] (80pt,30pt) -- (130pt,30pt)
node [pos=0.5,anchor=north,yshift=0pt] {\small{Even parity RR}};

\draw [
    thick,
    decoration={
        brace,
        mirror,
        raise=-15pt
    },
    decorate
] (50pt,-20pt) -- (80pt,-20pt)
node [pos=0.5,anchor=north,yshift=10pt] {\small{part of $U_i$}};

\draw [
    thick,
    decoration={
        brace,
        mirror,
        raise=-15pt
    },
    decorate
] (130pt,-20pt) -- (180pt,-20pt)
node [pos=0.5,anchor=north, xshift=10pt, yshift=10pt] {\small{part of $U_i$}};

\draw [
    thick,
    decoration={
        brace,
        mirror,
        raise=-15pt
    },
    decorate
] (70pt,-70pt) -- (160pt,-70pt)
node [pos=0.5,anchor=north,yshift=10pt] {\small{New RR}};

\node (cen) at (140pt, 35pt) {\small{midpoint of $\mathcal{U}_i$}};

\draw[dashed] (20pt,40pt) -- (20pt, -55pt);
\draw[dashed] (240pt,40pt) -- (240pt, -55pt);

\draw[dashed,color=red] (140pt,30pt) -- (140pt, -55pt);

\end{tikzpicture}          \caption{The placement of the cut in $U_i$ and the creation of the new region. The parts of $U_i$ are shown in blue and the midpoint of $\mathcal{U}_i$, where the cut is placed, is shown with a red dashed line. After the cut is placed, an equal amount of green (e.g., ``$+$'') and of red (e.g., ``$-$'') will be reserved to the left and the right of the position of the cut, ``jumping over'' RRs of even parity. The newly added parts of $U_i$ (which are now part of $R_i$ instead are shown in slightly different shades of their respective colors, for clarity. After the operation, the agent is $1/2$-satisfied, which is guaranteed by her value in the union of restricted regions $R_i$.}
        \label{fig:placecut}
    \end{figure}
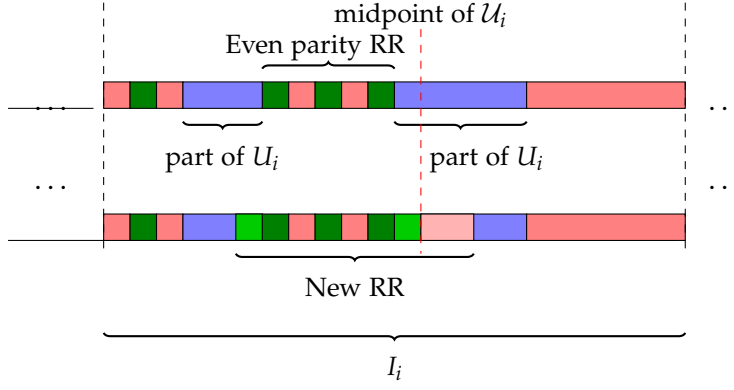

\subsubsection{Correctness of the algorithm} Below, we argue about the correctness of the algorithm. We will prove the correctness via a series of lemmas.

\begin{lemma}\label{lem:atmosthalf}
After step $i$ of the algorithm, the value of agent $i$ for any internal RR in $R_i$ is at most $1/2$.
\end{lemma}

\begin{proof}
We will prove the lemma using induction. For $i=1$, the statement is clearly true; there are no pre-existing RRs, so we place a cut in the midpoint $y_1$ of $I_1$ and reserve a region of total value $1/4$ for the agent to the left of $y_1$, and a region of total value $1/4$ to the right of $y_1$, which together form the first RR $R_1$. This is the only internal RR in $I_1$ and hence the lemma follows. 

Next, we will argue that the lemma holds for some $1 < j \leq n$. Consider agent $j$ and its value interval $I_j$. Since we consider agents in non-increasing order of the height of their valuation blocks, from the induction hypothesis, it holds that the value of agent $j$ for any RR which is internal for $I_i$, for any $i < j$, is at most $1/2$ as well. Therefore, it is not possible for agent $j$ to value any RR which is internal for $I_j$, at the beginning of step $j$, at more than $1/2$. During step $j$, the algorithm might either expand existing internal RRs, create a new RR, or both, but in any case, by construction, the union of these RRs will never amount to more than $1/2$ of agent $j$'s value. 
\end{proof}

\begin{algorithm}[t]
\caption{Polynomial-time algorithm for $1/2$-\ch} \label{alg:halfch}
  \KwIn{The value blocks of the agents, as pairs $(I_i, v_i)$, such that \begin{itemize}
      \item[-] $I_i=[\alpha_i,b_i]$ is the interval.
      \item[-] $v_i$ is the height of the value block.
  \end{itemize} }

  \KwOut{The set of cuts $c_1, \ldots, c_n$, which satisfies that $|\mu_i(I^{+}) - \mu_i(I^{-})\leq 1/2$ for every agent $i$.} \medskip
  
  Order the agents in non-increasing order of $v_i$ and let wlog $\{1,\ldots,n\}$ be any such ordering.\\
  \For{$i=1$ to $n$}{\medskip
  \tcc{Step (1), Expanding internal RRs of odd parity (\cref{sec:oddparity}).}
  Let $\mathcal{R^{\text{odd}}}$ be the set of \emph{internal} RRs of \emph{odd parity} in $I_i$ that can be expanded.\\ 
    \While{$\mathcal{R^{\text{odd}}} \neq \emptyset$  }{
        Pick $R$ in $\mathcal{R^{\text{odd}}}$ \\
        \textbf{Expand} $R$ until it can no longer be expanded, i.e., \\
~~~~~~~~~ -- The agent becomes $1/2$-satisfied. \\
             ~~~~~~~~~ -- The expansion reaches the endpoint of some other RR. \\
             ~~~~~~~~~ -- The expansion reaches the endpoint of $I_i$. \\
\If{Agent $i$ is $1/2$-satisfied}{ \smallskip
        Break and move on to the next agent.}
        Generate new RRs (if the expansion of $R$ is merged with some other RR).\\
        Update $\mathcal{R^{\text{odd}}}$.
    }\medskip
    \tcc{Step (2), Placing the cut in $U_i$  (see \cref{sec:placecut}).}
    Let $\mathcal{U}_i$ be the interval obtained by joining the intervals of $U_i$.\\
    Let $y$ be the midpoint of $\mathcal{U}_i$.\\
    Place cut $c_i$ on $y$.\\ \medskip
    \tcc{Step (3), Creating the new RR (see \cref{sec:virtualrr}).}
    Consider some $x \in \mathbb{R}$ and the intervals $[y-x,y], [y,y+x] \subseteq \mathcal{U}_i$.\\
    Pick $x$ to satisfy $v_i\cdot (y-x)+v_i\cdot (y+x)+\mu_i(R_i) = 1/2$.\\
    Let $y_1$ and $y_2$ be the points in $U_i$ corresponding to $y-x$ and $y+x$ respectively.\\
    Create a new RR $R'= [y_1, y_2]$ (possibly containing other RRs of even parity).  
    }
\end{algorithm}

\begin{lemma}\label{lem:boundary}
Let $R_b$ be any boundary RR in $I_i$ at step $i$. Then, it holds that $|\mu_i(R_b^{+}) - \mu_i(R_b^{-})| \leq 1/4$.
\end{lemma}

\begin{proof}
First, note that since $R_b$ is a boundary RR in $I_i$, it is part of some larger RR; let $R$ be that RR, i.e., $R_b \subseteq R$. Also, note that $R$ is ``balanced'', meaning that the total length of $R$ labeled ``$+$'' and the total length of $R$ labeled ``$-$'' are equal; this follows from the fact that in any of the steps (1) and (3) of the algorithm (\cref{sec:oddparity} and \cref{sec:virtualrr}) where the RRs are created or expanded, these operations are only done symmetrically with respect to the position of some cut. In other words, if $R$ is internal for some agent $j$, it holds that $\mu_{j}(R^{+}) = \mu_{j}(R^{-})$.

Assume by contradiction that $|\mu_i(R_b^{+}) - \mu_i(R_b^{-})| > 1/4$ and assume without loss of generality that the excess is in favor of ``$+$'' over ``$-$''. This in particular means that more than $1/4$ of the value of agent $j$ in $R_b$ is labeled ``$+$'', which means that $R_b$ contains an interval of length larger than $1/4\cdot v_i$ that is labeled ``$+$'', where $v_j$ is the height of agent $i$'s valuation block. Since $R$ is balanced, this means that $R$ must have length at least $1/2\cdot v_i$. However, consider the step $j$ when the $R$ was last modified (either created or expanded) before step $i$ and observe that $R$ was an internal RR for $I_j$ at the point. This means that agent $j$ had value larger than $v_{j}\cdot 1/(2\cdot v_i) \geq 1/2$ for $R$ at the point, contradicting \cref{lem:atmosthalf}.
\end{proof}

\begin{lemma}\label{lem:satisfied}
After step (3) of the algorithm (\cref{sec:virtualrr}), agent $i$ is $1/2$-satisfied.
\end{lemma}

\begin{proof}
Consider the imbalance between the portions of $I_i$ assigned to the two labels after step (1) of the algorithm, and assume without loss of generality that there is an excess of the label ``$+$'' over the label ``$-$'', i.e., $\mu_i(I_i^{+}) > \mu_i(I_i^{-})$, as otherwise the agent is perfectly satisfied and we are done. Since by definition, there can be at most $2$ external RRs for $I_i$, and since all internal RRs have a perfect balance of ``$+$'' and ``$-$'' for agent $i$, it follows from \cref{lem:boundary} that $\mu_i(R_i^{+}) - \mu_i(R_i^{-}) \leq 1/2$, i.e., the difference in value for the agent in the reserved regions of $I_i$ is at most $1/2$. If $\mu_i(R_i) \geq 1/2$, then we are done, as the agent is $1/2$-satisfied. Otherwise, we place the cut in step (2) of the algorithm on the midpoint $y$ of $U_i = I_i \setminus R_i$ and we reserve an equal portion of $I_i$ to the left and to the right of $y$, until we establish that $\mu_i(R_i)=1/2$.

The only thing left to argue is that the creation of the RR $R$ in step (3) of the algorithm that will make the agent $1/2$-satisfied is possible. First, observe that $R$ will only strictly contain RRs that are of even parity, as all RRs of odd parity were considered in step (1) of the algorithm and were either (a) transformed to RRs of even parity via merging or (b) expanded until their reached some endpoint of the interval, and therefore can not be strictly contained in $R$. From this, it follows that every sub-interval of $U_i$ that will be included in $R$ on the left side of the cut on $y$ will receive the same label (e.g., ``$+$'') and every sub-interval of $U_i$ that will be included in $R$ on the right side of the cut on $y$ will receive the opposite label (e.g., ``$-$''). 
\end{proof}

\begin{lemma}\label{lem:halfsat}
After any step $j \geq i$, agent $i$ is $1/2$-satisfied.
\end{lemma}

\begin{proof}
First, it follows straightforwardly from \cref{lem:satisfied} that agent $i$ is $1/2$-satisfied after step $i$. For any $j > i$, let $R_i(k)$ be the union of RRs of $I_i$ after step $k$ and consider $R_i(i)$ and $R_i(j)$. By the way RRs are constructed, and since they are never intersected by a new cut, it follows that $R_i(i) \subseteq R_i(j)$ and in particular, the balance of labels in $R_i(i)$ has remained unchanged in $R_i(j)$ (although the labels themselves might have the opposite parity). It follows that the agent is $1/2$-satisfied after step $j$.
\end{proof}

\noindent The correctness of the algorithm follows by applying \cref{lem:halfsat} for $j=n$.\\

\noindent Before we conclude the subsection, we mention that the algorithm can actually be applied to instances with piecewise constant valuations with $d$ blocks, as long as we have $d \cdot n$ cuts at our disposal. The idea is quite simple: Any agent that has (at most) $d$ value blocks will be replaced by (at most) $d$ agents, with single-block valuations. This requires the appropriate scaling of the heights of the blocks, to make sure that the functions are still probability measures. Then, we will run the algorithm, now on $n' \leq d \cdot n$ agents, and we will obtain a solution using $n'$ cuts. This set of cuts will also be a solution to the original instance, since the parameter $\varepsilon$ is constant ($1/2$) in both cases. We obtain the following corollary.

\begin{corollary}
There exists a polynomial-time algorithm for $1/2$-\ch, when the agents have \emph{piecewise constant} valuations with $d$ \emph{blocks} and we are allowed to use $d \cdot n$ cuts.
\end{corollary}

\subsection{A polynomial-time algorithm using \texorpdfstring{$2n - \ell$}{2n-l} cuts}
\label{sec:linearProgramming}

We conclude the section with the following theorem, stating that there is a polynomial-time
algorithm for solving the \emph{exact} \ch\ problem with single-block valuations, if we are allowed to use
$2n-\ell$ cuts, for any constant $\ell$.

\begin{theorem} \label{thm:linearProgramming}
Let $\ell$ be an integer constant. There exists a polynomial-time algorithm that given any single-block instance:
\begin{itemize}
    \item returns an \emph{exact} \ch\ solution that uses at most $2n-\ell$ cuts, or
    \item correctly outputs that no such solution exists.
\end{itemize}
\end{theorem}

\noindent Given the single-block valuations of the $n$ agents, we can determine numbers $0 \leq a_1 < a_2 < \dots < a_m < a_{m+1} \leq 1$ for some $m \in [2n-1]$ such that for every agent $i \in [n]$:
\begin{itemize}
    \item $\mu_i([0,a_1]) = \mu_i([a_{m+1},1]) = 0$, and
    \item for any $j \in [m]$, the density of $\mu_i$ in interval $(a_j,a_{j+1})$ is constant.
\end{itemize}
The upper bound $m \leq 2n-1$ comes from the fact that all agents have single-block valuations.

With this in hand, the case $\ell=1$ is very easy to solve. Indeed, it suffices to position a cut at the middle of each interval $[a_j,a_{j+1}]$, $j \in [m]$, and to alternate the labels $+$ and $-$ between the cuts. Then, the value of every agent within every interval is split in half, which implies that the same also holds for their value over the whole interval. Thus, this yields an exact \ch\ that uses $m \leq 2n-1$ cuts.

Now consider the case where $\ell$ is any integer constant. If $m \leq 2n-\ell$, the solution above uses at most $2n-\ell$ cuts and we are done. Thus, it remains to solve the case where $2n-\ell < m$. Let $k = 2n - \ell$. The algorithm proceeds by picking a subset $S \subseteq [m]$ of size $|S| = k$, and checking whether there exists a solution with one cut in each interval $[a_j,a_{j+1}]$, $j \in S$, by using a linear programming subroutine. Since the number of such subsets $S$ is $\binom{m}{k} = \binom{m}{m-k} \leq \mathcal{O}(n^\ell)$, where $\ell$ is a constant, the algorithm can try all possible subsets $S$ of size $k$.

Let us now describe how the linear programming subroutine for a given subset $S$ works. We are going to put one cut in each interval $[a_j,a_{j+1}]$ for $j \in S$, and alternate the labels $+$ and $-$ between the cuts. In other words, whenever we ``cross a cut'', the label changes. We assume that the left-most label is $+$. First of all, note that for any $j \in [m] \setminus S$, the whole interval $[a_j,a_{j+1}]$ will be assigned to the same label, and, in fact, we can determine which label, since we know the exact number of cuts to the left of that interval. Similarly, we also know to which label the intervals $[0,a_1]$ and $[a_{m+1},1]$ will be assigned. Thus, let $A_+$ denote the union of all intervals for which we already know that they will be assigned to label $+$. Define $A_-$ analogously for label $-$. Note that $[0,1] \setminus (A_+ \cup A_-) = \bigcup_{j \in S} [a_j,a_{j+1}]$ (where we ignore endpoints of intervals, since they do not matter for the valuations).

The LP will have one variable $x_j$ for each $j \in S$, and $x_j$ will represent the position of the cut in the interval $[a_j,a_{j+1}]$. Since we will alternate labels between cuts, we can partition $S = S_+ \cup S_-$, such that for any $j \in S_+$, the interval $[a_j,x_j]$ is labeled $+$ and the interval $[x_j,a_{j+1}]$ is labeled $-$. For any $j \in S_-$, the interval $[a_j,x_j]$ is labeled $-$ and the interval $[x_j,a_{j+1}]$ is labeled $+$.

We can now write down the LP, which, in fact, does not even have an objective function (any feasible solution will do).
\begin{align}\label{eq:LP-more-cuts}
&\mu_i(A_+) + \sum_{j \in S_+} \mu_i([a_j,x_j]) + \sum_{j \in S_-} \mu_i([x_j,a_{j+1}]) &   & \nonumber\\
&~~~~~~~~= \mu_i(A_-) + \sum_{j \in S_+} \mu_i([x_j,a_{j+1}]) + \sum_{j \in S_-} \mu_i([a_j,x_j]) & &\text{ for all } i \in [n]\\
&x_j \in [a_j,a_{j+1}] & &\text{ for all } j \in S \nonumber
\end{align}
Note that $\mu_i(A_+)$ and $\mu_i(A_-)$ are constants, while $\mu_i([a_j,x_j])$ and $\mu_i([x_j,a_{j+1}])$ can be expressed as affine linear functions of $x_j$, since every agent has constant density within $[a_j,a_{j+1}]$. Thus, all of the constraints can indeed be written as linear equations/inequalities. Note that the LP \eqref{eq:LP-more-cuts} can be constructed for any non-empty subset $S \subseteq [m]$. It is easy to see that any feasible solution $(x_j)_{j \in S}$ immediately yields a \ch\ that uses at most $|S|$ cuts. In particular, if $|S|=k$, then this yields a \ch\ that uses at most $k=2n-\ell$ cuts. The correctness of the algorithm follows from the claim below.

\begin{claim}
Let $c \in [m]$. If for all subsets $S \subseteq [m]$ with $|S|=c$, the LP \eqref{eq:LP-more-cuts} has no feasible solution, then the \ch\ instance does not admit a solution with at most $c$ cuts.
\end{claim}

\begin{proof}
Let $c_{min}$ denote the minimum number of cuts $c$ such that there exists a solution to the \ch\ instance that uses $c$ cuts. Note that $c_{min} \leq n$. The statement in the claim clearly holds for all $c < c_{min}$, since there indeed is no \ch\ solution that uses at most $c$ cuts. Furthermore, the statement also holds for $c=m$, since the (in that case, single) LP \eqref{eq:LP-more-cuts} admits a feasible solution, as explained above. We begin by proving that the statement holds for $c=c_{min}$. Then, we show that if it holds for some $c \geq c_{min}$, then it also holds for $c+1$, as long as $c+1 \leq m$. By induction, this proves the claim.

\paragraph{Base case $c=c_{min}$:}
Let $I^+,I^-$ denote an exact solution to the \ch\ instance that uses the minimum number of cuts possible for this instance, i.e., $c_{min}$ cuts. First of all, note that the intervals $[0,a_1]$ and $[a_{m+1},1]$ cannot contain any cut. Indeed, since all agents have value $0$ for these intervals, which lie at the edge of the cake, we could remove any such cut without destroying the perfect balance between $+$ and $-$. Since $c_{min}$ is the minimum number of cuts in any solution, this is impossible. Furthermore, all the cuts have to lie at distinct positions and the labels have to alternate between $+$ and $-$. Otherwise, it is easy to see that we can again remove cuts, without destroying the perfect balance.

At the beginning, all cuts are at the positions indicated by $I^+,I^-$ and they are \emph{unclaimed}. For all $j \in [m]$, starting with $j=1$ and going up to $j=m$, we perform the following operations to construct $S$ and $(x_j)_{j \in S}$. Consider the interval $[a_j,a_{j+1}]$.
\begin{itemize}
    \item If the interval does not contain any unclaimed cut, then move on to the next $j$.
    \item If the interval contains exactly one unclaimed cut, at some position $s$, then mark the cut as claimed, add $j$ to the set $S$, let $x_j := s$, and move on to the next $j$.
    \item If the interval contains exactly two unclaimed cuts, at positions $s_1 < s_2$, then move the first cut to $a_{j+1}-(s_2-s_1)$ and claim it, and move the second cut to $a_{j+1}$. Then, add $j$ to the set $S$, let $x_j := a_{j+1}-(s_2-s_1)$ and move on to the next $j$. Note that the distance between the two cuts is still $s_2-s_1$, and since both cuts also still lie in $[a_j,a_{j+1}]$, where all agents have constant density, we keep the perfect balance between $+$ and $-$.
\end{itemize}
No other case can occur. In particular, it is not possible that $s_1=s_2$ in the third bullet point, since that would mean that the two cuts can be removed without destroying the perfect balance. Furthermore, it is not possible for the interval to contain more than two cuts (claimed or unclaimed), since then we can definitely remove at least two cuts without destroying the perfect balance. Indeed, if there are three cuts, then it is easy to see that they can be replaced by a single cut. If there are four cuts, then they can be replaced by two cuts.

When we reach $j=m$, the interval $[a_m,a_{m+1}]$ can contain at most one cut. If it is unclaimed, it will be claimed by $x_m$ and we will add $m$ to the set $S$. If the interval contained more than one cut, then using the same arguments as above we could either remove one cut, or, if there were two cuts, we could shift them to the right until the second cut reaches $a_{m+1}$, but then this cut can again be removed.

Note that after we have finished our pass, all cuts must have been claimed. Thus, we obtain a set $S \subseteq [m]$ of size $|S|=c_{min}$ and values $(x_j)_{j \in S}$ that satisfy the LP \eqref{eq:LP-more-cuts}. This proves that the statement of the claim holds for $c=c_{min}$.

\paragraph{From $c$ to $c+1$:} Consider any $c \geq c_{min}$ such that the statement of the claim holds and $c+1 \leq m$. Since the statement holds for $c$ and, by definition of $c_{min}$, there exists a solution to the \ch\ instance that uses at most $c$ cuts, it follows that there exists a set $S \subseteq [m]$ with $|S| = c$, and values $(x_j)_{j \in S}$ that satisfy the LP \eqref{eq:LP-more-cuts}.

We show how to construct a set $S' \subseteq [m]$ with $|S'| = c+1$, and values $(x_j)_{j \in S'}$ that satisfy the LP \eqref{eq:LP-more-cuts}. It immediately then follows that the statement of the claim also holds for $c+1$.

At the beginning, all cuts are at the positions indicated by $(x_j)_{j \in S}$ and they are \emph{claimed}. We add an additional unclaimed cut at position $a_1$. Note that we now have $c+1$ cuts and they still form a solution to the \ch\ instance. For all $j \in [m]$, starting with $j=1$ and going up to $j=m$, we perform the following operations. Consider the interval $[a_j,a_{j+1}]$. The unclaimed cut is at position $a_j$.
\begin{itemize}
    \item If $j \notin S$, then let $S' = S \cup \{j\}$, set $x_j := a_j$, and terminate.
    \item If $j \in S$, then update $x_j := a_j + (a_{j+1}-x_j)$, put the unclaimed cut at position $a_{j+1}$ and move on the next $j$. Note that the cuts still form a solution to the \ch\ instance, since the distance between the two cuts in the interval has not changed.
\end{itemize}
Since $c \leq m-1$, there exists $j \in [m] \setminus S$, and the procedure will terminate when it reaches the smallest such $j$. It is easy to see that $S'$ and $(x_j)_{j \in S'}$ satisfy the LP \eqref{eq:LP-more-cuts}.
\end{proof} 

\section{\textsc{Consensus-\texorpdfstring{$1/3$}{1/3}-Division} is PPAD-hard}\label{sec:ppad3hardness}

As we mentioned in the Introduction, our newly developed tools that allowed us to obtain a strengthening of the PPA-completeness result for \ch, turn out to be very useful for proving a hardness result for a more general version of the problem, the \cd\ problem, for $k=3$. We provide the definition below.

\begin{definition}[$\varepsilon$-\cd]
Let $k \geq 2$. We are given $\varepsilon > 0$ and continuous probability measures $\mu_1, \dots, \mu_n$ on $[0,1]$. The probability measures are given by their density functions on $[0,1]$. The goal is to partition the unit interval into $k$ (not necessarily connected) pieces $A_1,\dots,A_k$ using at most $(k-1)n$ cuts, such that $|\mu_j(A_i) - \mu_j(A_\ell)| \leq \varepsilon$ for all $i,j,\ell$.
\end{definition}

\noindent For $\varepsilon$-\textsc{Consensus-$1/3$-Division}, for ease of notation, we will use the labels A/B/C instead of $1/2/3$. We state the main theorem of the section.

\begin{theorem}\label{thm:ppad3hard}
$\varepsilon$-\textsc{Consensus-$1/3$-Division} is PPAD-hard, for inverse exponential approximation parameter $\varepsilon$.
\end{theorem}

As we discussed in the Introduction, the problem for $k\geq 3$ is not necessarily harder than the case of $k=2$, because we have more cuts at our disposal. Before we present the proof, we highlight some fundamental challenges that arise when one moves from $k=2$ to larger $k$.

First, when moving to $k \geq 3$, we move from a simple $+/-$ parity to a more general setting with at least $3$ labels. This is already severely problematic when it comes to the reduction of \citep{FRG18-Necklace}, which is highly dependent on the solution having two labels. Indeed, the interpretation of the inputs in \citep{FRG18-Necklace} and the corresponding design of the gates needs to know which label will appear on the left side of the cut, and special ``parity-flip`` gadgets are used throughout the reduction to ensure this. On the contrary, with our new value interpretation, we design gadgets which take the label sequence ``internally'' into account, by adjusting the position of the cut accordingly. 

The sequence of the labels however gives rise to a second and more challenging issue: While in the case of $k=2$ we can assume without loss of generality that the labels between cuts alternate between ``$+$'' and ``$-$'', we cannot make any such assumptions even when $k=3$. In fact, it is known that if one restricts the solution to exhibit a cyclic sequence of labels $A/B/C$, then the problem is no longer a total search problem \citep{palvolgyi2009combinatorial}. This seems to be a fundamental obstacle to the design of gates for the case of $k \geq 3$. For $k=3$, we manage to side-step this obstacle by using a clever ``trick'': we make sure that the intervals of the \textsc{Consensus-$1/3$-Division} instance where we read the two inputs (of the $2$-dimensional PPAD-complete problem that we reduce from, see \cref{def:fixp}) are placed next to each other, therefore fixing the position of one of the three labels. We prove PPAD-hardness for the \emph{exact} version of the problem (in which we are looking for a perfect balance of the labels, with no allowable discrepancy $\varepsilon$), which will guarantee that in a solution, the value of this label will be fixed to $1/3$ throughout the instance. Since our instance is constructed to have piecewise constant valuation functions, the result can be extended to the case of inverse-exponential $\varepsilon$ using the following lemma, which is based on an argument of \citet{EY10-Nash-FIXP}.

\begin{lemma}\label{lem:exact-to-approx}
Let $k \geq 2$. For piecewise constant valuations, exact \cd\ reduces in polynomial time to $\varepsilon$-\cd\ with inverse-exponential $\varepsilon$.
\end{lemma}

\begin{proof}
We use the proof idea introduced by \citet{EY10-Nash-FIXP}. Given an instance of Consensus-$1/k$-Division with piecewise constant valuations that we want to solve exactly, we first find an $\varepsilon$-approximate solution $S$ for some sufficiently small $\varepsilon$. Then, using $S$, we construct an LP in polynomial time. Finally, we show that any solution to this LP yields an exact solution to the Consensus-$1/k$-Division instance.

Let $I$ be an instance of Consensus-$1/k$-Division with agents $\{1,2, \dots, n\}$ that have piecewise constant valuations. Then, we can efficiently find positions $0 = p_0 < p_1 < \dots < p_{m-1} < p_m = 1$ such that for all agents $i \in [n]$ and all $j \in [m]$, the density of the valuation function $\mu_i$ in interval $[p_{j-1},p_{j}]$ is constant. Denote that constant by $f_{i,j}$. Note that $m$ and the bit-length of $f_{i,j}$ are polynomially upper bounded by the size of the instance $I$.

Let $S = (\widehat{x}, \textup{label})$ be a candidate solution, where the cut positions are $0 = \widehat{x}_0 \leq \widehat{x}_1 \leq \widehat{x}_2 \leq \dots \leq \widehat{x}_{(k-1)n-1} \leq \widehat{x}_{(k-1)n} \leq \widehat{x}_{(k-1)n+1} = 1$ and for $t \in [(k-1)n+1]$, $\textup{label}(t)$ is the label in $[k]$ assigned to interval $[\widehat{x}_{t-1},\widehat{x}_t]$.

We construct an LP that has variables $x_0, \dots, x_{(k-1)n+1}$ and a variable $z$. We denote it by $\textup{LP}(I,S)$. In the LP we minimize $z$ under the constraints:
\begin{itemize}
    \item $x_0=0$ and $x_{(k-1)n+1}=1$ (these variables are only used for notational convenience)
    \item for every $t \in [(k-1)n]$: cut $x_t$ must lie between cuts $x_{t-1}$ and $x_{t+1}$
    \item for every $t \in [(k-1)n]$: for every $j$ such that $\widehat{x}_t \in [p_{j-1},p_j]$ (at most two such $j$ exist), $x_t$ must also lie in $[p_{j-1},p_j]$
    \item let $A_1, \dots, A_k$ denote the partition of the domain $[0,1]$ given by the cuts $x_0, \dots, x_{(k-1)n+1}$ and the labeling $\textup{label}$. Then the constraint is:
    
    for every agent $i \in [n]$: $\max_{\ell_1,\ell_2} |\mu_i(A_{\ell_1})-\mu_i(A_{\ell_2})| \leq z$
\end{itemize}
We could also add a constraint $z \geq 0$, but this is already implicitly enforced by the existing constraints.

Apart from the last constraint type, it is clear that all the other constraints can be expressed linearly. For the last type of constraint, note that it can be broken down into all the constraints $\mu_i(A_{\ell_1})-\mu_i(A_{\ell_2}) \leq z$ and $\mu_i(A_{\ell_2}) - \mu_i(A_{\ell_1}) \leq z$ for all $\ell_1, \ell_2$. Thus, it remains to show that $\mu_i(A_{\ell_1})$ and $\mu_i(A_{\ell_2})$ can be written as linear functions of $x_0, \dots, x_{(k-1)n+1}$.

To see this, first note that the labeling is fixed, i.e.\ for any interval $[x_t,x_{t+1}]$ we know which label it is assigned to. Furthermore, for every $t \in [(k-1)n]$ the cut $x_t$ is constrained to lie in some interval $[p_{j-1},p_j]$, and it is not allowed to cross from one interval $[p_{j-1},p_j]$ to another. Thus, for any agent $i \in [n]$ and for any interval $[p_{j-1},p_j]$, we can express the amount of value of agent $i$ in $[p_{j-1},p_j]$ going to each of the labels $\{1,2, \dots, k\}$ as a linear expression. If $[p_{j-1},p_j]$ does not contain any cut $x_t$, then all of it is assigned to some label $\ell$ (which is fixed, i.e.\ only depends on $S$). Thus, the value assigned to label $\ell$ is $f_{ij}(p_j-p_{j-1})$, and the value assigned to all other labels is $0$. If $[p_{j-1},p_j]$ contains the cuts $x_{t_1} \leq \dots \leq x_{t_s}$, then:
\begin{itemize}
    \item interval $[p_{j-1},x_{t_1}]$ yields value $f_{ij}(x_{t_1}-p_{j-1})$ to $\textup{label}(t_1)$ (and value $0$ to all other labels)
    \item interval $[x_{t_p},x_{t_{p+1}}]$ (for $1 \leq p \leq s-1$) yields value $f_{ij}(x_{t_{p+1}}-x_{t_p})$ to $\textup{label}(t_{p+1})$
    \item interval $[x_{t_s},p_j]$ yields value $f_{ij}(p_j-x_{t_s})$ to $\textup{label}(t_s+1)$
\end{itemize}
Note that any feasible solution to this LP with $z=0$ is an exact solution to our Consensus-$1/k$-Division instance $I$. Furthermore, note that there exist polynomials $p_1$ and $p_2$ such that for any candidate solution $S$, the number of equations in $\textup{LP}(I,S)$ is at most $p_1(|I|)$ and the bit-size of any number appearing in $\textup{LP}(I,S)$ is at most $p_2(|I|)$. Here $|I|$ denotes the representation size of the input $I$. Thus, there exists some polynomial $q$ such that we can efficiently compute an optimal solution of $\textup{LP}(I,S)$ where the bit-size of the variables is at most $q(|I|)$. Note that this does not depend on $S$.

Now assume that we have picked $\varepsilon < 1/2^{q(|I|)}$ and obtain a solution $S = (\widehat{x}, \textup{label})$ to $\varepsilon$-Consensus-$1/k$-Division($I$). We then construct $\textup{LP}(I,S)$ and compute an optimal solution $(x^*,z^*)$. Note that $(\widehat{x},\varepsilon)$ is a feasible solution to $\textup{LP}(I,S)$. Thus, it must hold that $z^* \leq \varepsilon < 1/2^{q(|I|)}$. But since the bit-size of $z^*$ is at most $q(|I|)$ and $z^* \geq 0$, it must hold that $z^* = 0$. Thus, $S'=(x^*, \textup{label})$ is an exact solution to our Consensus-$1/k$-Division instance $I$.
\end{proof}

\subsection{A problem to reduce from: \tlinfixp}

We start with the following problem, proven to be PPAD-complete by \citet{mehta2014constant}.
\begin{definition}[\linfixp~\citep{mehta2014constant}]
The problem \linfixp\ is defined as follows. We are given a circuit $C$ using gates $\{+, \times \zeta, \max\}$ and rational constants, that computes a function $F_C : [0,1]^2 \to [0,1]^2$. The goal is to find $x \in [0,1]^2$ such that $F_C(x)=x$.
\end{definition}

\begin{theorem}[\citep{mehta2014constant}]
\linfixp\ is \textup{PPAD}-complete.
\end{theorem}

In order to prove PPAD-hardness of Consensus-1/3-Division, we will use a slightly modified version of \linfixp, that we call \tlinfixp. The first difference is that the domain is $[-1,1]^2$ instead of $[0,1]^2$. Furthermore, instead of the gates $\{+, \times \zeta, \max\}$ and rational constants in $\mathbb{Q}$, the circuit will only be allowed to use the gates $\{ \tplus, \ttimes \zeta\}$ and rational constants in $[-1,1] \cap \mathbb{Q}$. The gate $\tplus$ corresponds to \emph{truncated addition} and the gate $\ttimes \zeta$ corresponds to \emph{truncated multiplication by} $\zeta$. For any $x,y \in [-1,1]$ and any $\zeta \in \mathbb{Q}$, we have $x \tplus y = T[x+y]$ and $x \ttimes \zeta = T[x \times \zeta]$, where $T$ is the truncation operator in $[-1,1]$ as defined in \cref{sec:ppahardness}.

\begin{definition}[\tlinfixp]\label{def:fixp}
The computational problem \tlinfixp\ is defined as follows. We are given a circuit $C$ using gates $\{ \tplus, \ttimes \zeta\}$ and rational constants in $[-1,1]$, that computes a function $F_C : [-1,1]^2 \to [-1,1]^2$. The goal is to find $x \in [-1,1]^2$ such that $F_C(x)=x$.
\end{definition}

By applying some simple modifications to any \linfixp\ circuit, we are able to show the following theorem. 
\begin{theorem}\label{thm:fixp}
\tlinfixp\ is \textup{PPAD}-complete.
\end{theorem}

\begin{proof}
Containment in PPAD follows from the containment of the more general problem \textsc{Linear-FIXP} in PPAD \citep{EY10-Nash-FIXP}. In particular, note that the truncated gates can be simulated by using their non-truncated versions, along with $\max$ and constant-gates.

To show PPAD-hardness, we reduce from \linfixp. Let $C$ be a circuit that computes a function $F_C: [0,1]^2 \to [0,1]^2$ that uses gates $\{+,\max,\times \zeta\}$ and rational constants. First, let us change the domain to $[-1,1]^2$. We modify the circuit $C$ into $\widehat{C}$ by applying the following transformation (also used in \citep{mehta2014constant}) to every output: $x \mapsto \max\{\min\{1,x\},0\} = \max\{-1\times\max\{-1,-1 \times x\},0\}$. Then, for any input $(x,y) \in [-1,1]^2$, we must have that $F_{\widehat{C}}(x,y) \in [0,1]^2$ and for $(x,y) \in [0,1]^2$ we have $F_{\widehat{C}}(x,y) = F_C(x,y)$. It follows that $\widehat{C}$ computes a function $F_{\widehat{C}} : [-1,1]^2 \to [-1,1]^2$ and any fixed-point of $F_{\widehat{C}}$ must be in $[0,1]^2$ and thus also a fixed-point of $F_C$. This means that without loss of generality, we can assume that the domain is $[-1,1]^2$ instead of $[0,1]^2$.

Next, let us show how to replace the gates $\{+, \times \zeta\}$ by $\{\tplus, \ttimes \zeta\}$, and ensure that the constant-gates all have value in $[-1,1]$. Let $c \geq 2$ be an upper bound on the absolute value of all the constants appearing in the circuit, i.e.\ both the constant-gates and the $\times \zeta$-gates. Note that $c$ has bit representation length that is polynomial in the size of $C$ (since the size of $C$ includes the representation length of all constants). Let $n$ be the number of gates $\{+,\max,\times \zeta\}$ in $C$.

Begin with the circuit $C_0$ that only contains the two input gates of $C$ and all the constant-gates of $C$. Pick an arbitrary gate in $C$ with input(s) in $C_0$ and add it to $C_0$ to obtain circuit $C_1$. Then, pick an arbitrary gate in $C$ (and not yet in $C_1$) with input(s) in $C_1$ and add it to $C_1$ to obtain $C_2$. If we repeat this $n$, we obtain an incremental sequence of circuits $C_0$, $C_1$, $\dots$, $C_n = C$, where a single gate is added at each step.

For $i \in \{0,1,\dots,n\}$, let $v(C_i)$ denote the maximum absolute value of any gate in $C_i$, over all inputs in $[-1,1]^2$. Clearly, we have $v(C_0) \leq \max\{1,c\} = c$. To obtain $C_1$ from $C_0$, we have added a single gate. If this gate is a $+$-gate, then maximum absolute value appearing in the circuit is at most multiplied by $2 \leq c$. If it is a $\max$-gate, then the maximum absolute remains the same. Finally, if it is a $\times \zeta$-gate, the maximum absolute value is at most multiplied by $|\zeta| \leq c$. Thus, in any case, we have $v(C_0) \leq c v(C_1)$. By induction it follows that $v(C) = v(C_n) \leq c^n v(C_0) = c^{n+1}$. Note that $M := c^{n+1}$ has representation length that is polynomial with respect to the size of $C$.

From the arguments above, it follows that if the input is in $[-1,1]^2$, all the intermediate values in the computation of the circuit are upper bounded by $M$ in absolute value. Now modify $C$ to obtain $C'$ as follows. For every constant-gate, replace its constant $\zeta$ by $\zeta/M$. For every input, introduce a $\times (1/M)$-gate that multiplies it by $1/M$ and then use the output of that gate as the corresponding input for $C$. By induction, it is easy to see that on any input $(x,y) \in [-1,1]^2$, $C'$ computes $F_C(x,y)/M$. Importantly, all the intermediate values in the computation of the circuit are now upper bounded by $1$ in absolute value. In particular, all the constant-gates have value in $[-1,1]$.

For any $(x,y) \in [-1,1]^2$, we know that $F_C(x,y) \in [-1,1]^2$, i.e.\ $F_C(x,y)/M \in [-1/M,1/M]^2$. Obtain circuit $C''$ by taking $C'$ and multiplying all the outputs by $M$ (by using $\times M$-gates). Then, on any input $(x,y) \in [-1,1]^2$, $C''$ outputs $F_C(x,y)$ and all intermediate values in the computation of the circuit are upper bounded by $1$ in absolute value.

Thus, we have transformed $C$ into a circuit that computes the same function $F_C: [-1,1]^2 \to [-1,1]^2$, but only uses gates $\{\tplus, \ttimes \zeta, \max\}$, and all the constant-gates have value in $[-1,1]$.

The last step is to get rid of $\max$-gates. To do this observe that for any $x,y \in [-1,1]$, we have
$$\max\{x,y\} = \left(\frac{1}{2}\ttimes x + \max\left\{\frac{1}{2}\ttimes y \tplus \left(-\frac{1}{2}\right) \ttimes x, 0\right\}\right) \ttimes 2$$
and
$$\max\{x,0\} = (x \tplus (-1)) \tplus 1$$
Using these two equations we can implement any $\max$-gate by using the gates $\{\tplus, \ttimes \zeta\}$ and rational constants in $[-1,1]$.

Putting everything together, we have provided a reduction from \linfixp\ to \tlinfixp.
\end{proof}

\subsection{Description of the reduction}

In order to show that Consensus-1/3-Division is PPAD-hard, we reduce from \tlinfixp. Namely, given a \tlinfixp\ circuit $C$, we will construct an instance $I_C$ of Consensus-1/3-Division, such that any solution of $I_C$ yields a solution to $C$ (i.e., a fixed-point of $F_C: [-1,1]^2 \to [-1,1]^2$). As before, we will construct a Consensus-1/3-Division instance on some domain $[0,M]$ for some polynomial $M$, which is easy to transform into an equivalent instance on domain $[0,1]$.

First, let us give a high-level description of the \emph{ideal} reduction that we would like to construct. First, we would show how any interval of the Consensus-1/3-Division domain encodes a value in $[-1,1]$. Namely, in any solution $S$ to instance $I_C$, for every interval $I$, $v_{S}(I) \in [-1,1]$ would be the value encoded by interval $I$. Then, we would construct agents that implement the arithmetic gates needed by \tlinfixp, namely $\{\tplus, \ttimes \zeta\}$ and rational constants in $[-1,1]$. These agents read some value(s) in $[-1,1]$ from one or two intervals and output the result of the gate-operation into some other interval of the domain.

With these gates we could implement the circuit $C$ inside our Consensus-1/3-Division instance. In particular, we would have two intervals $In_1$ and $In_2$ each representing the two inputs, and two intervals $Out_1$ and $Out_2$ each representing the two outputs. These intervals are pairwise disjoint. In the final step we would then ``connect'' the outputs to the inputs. Namely, we would introduce an agent implementing a $\ttimes 1$-gate with input $Out_1$ and output $In_1$, and a second agent implementing a $\ttimes 1$-gate with input $Out_2$ and output $In_2$. This ensures that from any solution of the Consensus-1/3-Division instance we can extract a fixed-point of $F_C$.

If we could do all this, then this reduction would be very similar to the reduction of \citet{filos2018hardness} showing that $\varepsilon$-Consensus-Halving is PPAD-hard for constant $\varepsilon$. Unfortunately, there is a significant obstacle. Namely, we don't know how to find an encoding of values in intervals such that we can implement arithmetic gates that always work. Because we don't know in what order the labels $A$, $B$ and $C$ will appear in any given interval, implementing arithmetic gates is actually much harder than in the case of Consensus-Halving. Thus, the gates we are able to implement only work if the input interval encodes a value in a very specific way. In this case, we say that the interval is a \emph{valid encoding} of a value. Not all intervals will be a valid encoding of a value. In general, it is very hard to enforce valid encodings. This is the reason why our reduction does not seem to generalize to yield hardness for inversely polynomial $\varepsilon$, or to Consensus-$1/k$-Division with $k > 3$.

Nevertheless, for exact Consensus-1/3-Division we are able to find a work-around to force all intervals to be valid encodings of a value. If an interval is a valid encoding of a value (in solution $S$), then let $v_S(I) \in [-1,1]$ denote the value encoded by $I$. We will drop the subscript $S$ in the remainder of the exposition. The following two lemmas are crucial. They are proved in \cref{sec:ppad-details}, where the proof of \cref{thm:ppad3hard} is detailed.

\begin{lemma}\label{lem:ppad-projection}
In the instance $I_C$ we construct, it holds that:
\begin{itemize}
    \item the two intervals $In_1$ and $In_2$ are valid encodings, and
    \item if $Out_1$ and $Out_2$ are valid encodings, then $v(Out_1)=v(In_1)$ and $v(Out_2)=v(In_2)$.
\end{itemize}
\end{lemma}

\begin{lemma}\label{lem:ppad-gates}
In instance $I_C$, we can implement arithmetic gates $\{G_{\tplus}, G_{\ttimes \zeta}, G_\zeta\}$ for operations $\{\tplus, \ttimes \zeta\}$ and constant-gate respectively, such that:
\begin{itemize}
    \item $G_\zeta$ outputs a valid encoding of $\zeta \in [-1,1] \cap \mathbb{Q}$
    \item if the input to $G_{\ttimes \zeta}$ is a valid encoding of value $x \in [-1,1]$, then the gate outputs a valid encoding of $x \ttimes \zeta$
    \item if the two inputs to $G_{\tplus}$ are valid encodings of $x,y \in [-1,1]$, then the gate outputs a valid encoding of $x \tplus y$.
\end{itemize}
\end{lemma}

If these two lemmas indeed hold, then the reduction is correct. First of all, by \cref{lem:ppad-projection}, the two inputs to the circuit are valid encodings. Thus, using \cref{lem:ppad-gates}, it follows by induction that all the gates in the circuit perform their operation correctly and output a valid encoding. In particular, the two outputs of the circuit are valid encodings. Thus, by \cref{lem:ppad-projection}, we get that each of the two outputs is equal to the corresponding input. As a result, we have identified a fixed-point of $F_C: [-1,1]^2 \to [-1,1]^2$.

\subsection{Details of the proof}\label{sec:ppad-details}

We construct an instance with a set of agents $\{1,2,\dots, n\}$ for some $n$ that is polynomial in the size of $C$ (the exact value of $n$ can easily be inferred from the rest of the proof).
For every $i \in [n]$, agent $i$ will have an \emph{output interval} $O_i$ of length $9$ on the domain with the following properties:
\begin{itemize}
    \item for all $j \neq i$ it holds $O_j \cap O_i = \emptyset$ (output intervals are pairwise disjoint)
    \item agent $i$'s density function has height $3/10$ in $O_i[1,2] \cup O_i[4,5] \cup O_i[7,8]$.
\end{itemize}
where the notation $O_i[a,b]$ is used to denote a sub-interval of $O_i$, and $O_i[0,9] = O_i$.

From the second point, it follows that in any solution $S$, for each $i \in [n]$, $O_i$ must contain at least two cuts. Indeed, since the interval $O_i$ contains at least $3 \times 3/10 = 9/10$ value for agent $i$ (out of the total $1$), it is impossible for that agent to be satisfied unless there are at least two cuts in $O_i$. From the first point, namely the pairwise disjointness of the intervals $O_i$, it follows that all $2n$ cuts are accounted for and thus for each $i \in [n]$, $O_i$ contains exactly two cuts. But then it is easy to check that $O_i$ must be \emph{well-cut}: one cut lies in $O_i[7/4,17/4]$ and the other one in $O_i[19/4,29/4]$. In particular, there are no stray cuts in this reduction.

\paragraph{\textbf{Representation of a value.}} In any solution $S$, every interval $I$ of length 9 encodes a value as follows. Let $X(I) = I[0,1] \cup I[2,4] \cup I[5,7] \cup I[8,9]$. Let $X(I)_A$, $X(I)_B$ and $X(I)_C$ denote the subsets of $X(I)$ allocated by the solution $S$ to labels $A$, $B$ and $C$ respectively. Let $\mu$ denote the Lebesgue measure on $\mathbb{R}$. We say that $I$ is a \emph{valid encoding} of a value, if $\mu(X(I)_C) = 2$ and if $I$ is well-cut. If $I$ is a valid encoding, then the encoded value is $v(I) = (\mu(X(I)_A)-\mu(X(I)_B))/2 \in [-1,1]$.

In the next section, we show how to construct the arithmetic gates. With these gates, we can implement the circuit $C$ in our instance $I_C$. We will implement the circuit such that the first output of the circuit is encoded in the left-most interval of the domain $Out_1 = [0,9]$, and the second output of the circuit is encoded in the next interval, i.e.\ $Out_2 = [10,19]$. We will ensure that the next two intervals are not used by any gate of the circuit, namely $Temp_1=[20,29]$ and $Temp_2=[30,39]$ are available. Finally, the next two intervals correspond to the two inputs of the circuit, i.e.\ $In_1=[40,49]$ and $In_2=[50,59]$.

By construction, we know that $Out_1$ is well-cut, because it is the output interval of some agent. In particular, it contains exactly two cuts. By convention, we will assume that the labels appear in order $A$, $B$, $C$ from left to right in $Out_1$. Given any solution $S$ of $I_C$, we can always ensure that this holds by renaming the labels. Since $Out_1$ is well-cut, we know that interval $Out_1[0,1/2]$ is labeled $A$, interval $Out_1[17/4,19/4]$ is labeled $B$, and interval $Out_1[17/2,9]$ is labeled $C$. Furthermore, since interval $Out_2$ is also well-cut and right next to $Out_1$, we know that $Out_2[0,1/2]$ is labeled $C$. This observation is crucial for this reduction. However, note that we do not know the labels of $Out_2[17/4,19/4]$ and $Out_2[17/2,9]$, except that one of them is $A$ and the other $B$ (but not which one is which).

\paragraph{\textbf{Projection Agents.}} In order to ensure that \cref{lem:ppad-projection} holds, we are going to introduce two special agents, that we call \emph{projection agents}. Let $I := Out_1$ and $O := Temp_1$. The first projection agent's valuation function has height $3/10$ in $O[1,2] \cup O[4,5] \cup O[7,8]$, height $1/120$ in $X(O)$, height $1/30$ in $I[17/2,9]$, height $1/120$ in $I[2,4]$, height $1/60$ in $I[0,1/2] \cup I[17/4,19/4]$, and height $0$ everywhere else. See \cref{fig:projection-gadget} for an illustration of the first projection agent's valuation function. Since $Out_1$ is well-cut and we know the labels of the pieces, this agent ensures that:
\begin{itemize}
    \item interval $Temp_1$ is a valid encoding
    \item if interval $Out_1$ is a valid encoding, then $v(Temp_1) = - v(Out_1)$
\end{itemize}
The first property holds because exactly $1/3$ of the projection agent's value in interval $Out_1$ goes to label $C$. Thus, label $C$ also needs to get exactly $1/3$ of the agent's value in interval $Temp_1$. By construction of the valuation function in $Temp_1$, it follows that $C$ must get exactly $1/3$ of $X(Temp_1)$. To show the second property, consider the three possible cases:
\begin{itemize}
    \item label $C$ is the right-most label in $Temp_1$: then there is a cut at position $Temp_1[6]$. In this case, it is easy to see that the other cut in $Temp_1$ will exactly reflect what the first cut in $Out_1$ does.
    \item label $C$ is the middle label in $Temp_1$: since $C$ gets $1/3$ of $X(Temp_1)$, it follows that the cuts in $Temp_1$ have distance exactly $3$ between them. From this, it follows that both cuts reflect what the first cut in $Out_1$ does.
    \item label $C$ is the left-most label in $Temp_1$: then there is a cut at position $Temp_1[3]$. Again, as in the first case, the other cut in $Temp_1$ will exactly reflect what the first cut in $Out_1$ does. (Note that this case is actually not possible here, but we have included it, because it might occur for the second projection agent).
\end{itemize}

\begin{figure}
\centering
\begin{tikzpicture}[scale=0.6,transform shape]

	\draw (0,0) -- (9,0);
	\draw[dotted] (9,0) -- (11,0);
	\draw (11,0) -- (20,0);
	
	\draw (0,0) -- (0,-0.2);
	\draw (1,0) -- (1,-0.2);
	\draw (2,0) -- (2,-0.2);
	\draw (3,0) -- (3,-0.2);
	\draw (4,0) -- (4,-0.2);
	\draw (5,0) -- (5,-0.2);
	\draw (6,0) -- (6,-0.2);
	\draw (7,0) -- (7,-0.2);
	\draw (8,0) -- (8,-0.2);
	\draw (9,0) -- (9,-0.2);
	
	\node at (4.5,-1.5) {\Large{$I := Out_1$}};
	
	\draw (11,0) -- (11,-0.2);
	\draw (12,0) -- (12,-0.2);
	\draw (13,0) -- (13,-0.2);
	\draw (14,0) -- (14,-0.2);
	\draw (15,0) -- (15,-0.2);
	\draw (16,0) -- (16,-0.2);
	\draw (17,0) -- (17,-0.2);
	\draw (18,0) -- (18,-0.2);
	\draw (19,0) -- (19,-0.2);
	\draw (20,0) -- (20,-0.2);
	
	\node at (15.5,-1.5) {\Large{$O := Temp_1$}};
	
	\draw[fill=red!20!white] (12,0) rectangle (13,18/5);
	\draw[fill=red!20!white] (15,0) rectangle (16,18/5);
	\draw[fill=red!20!white] (18,0) rectangle (19,18/5);
	\draw[fill=red!20!white] (11,0) rectangle (12,1/10);
	\draw[fill=red!20!white] (13,0) rectangle (15,1/10);
	\draw[fill=red!20!white] (16,0) rectangle (18,1/10);
	\draw[fill=red!20!white] (19,0) rectangle (20,1/10);
	
	\draw[fill=red!20!white] (2,0) rectangle (4,1/10);
	\draw[fill=red!20!white] (0,0) rectangle (0.5,2/10);
	\draw[fill=red!20!white] (17/4,0) rectangle (19/4,2/10);
	\draw[fill=red!20!white] (17/2,0) rectangle (9,4/10);

	\draw[dashed,color=blue] (2.7,2) -- (2.7, -0.5);
	\draw[dashed,color=blue] (7.5,2) -- (7.5, -0.5);
	
	\node[color=blue] at (1.3,1.3) {\Large{$A$}};
	\node[color=blue] at (5.3,1.3) {\Large{$B$}};
	\node[color=blue] at (8.3,1.3) {\Large{$C$}};
	
	\draw[dashed,color=blue] (13.7,2) -- (13.7, -0.5);
	\draw[dashed,color=blue] (17,2) -- (17, -0.5);
	
	\node[color=blue] at (11.5,1.3) {\Large{$B$}};
	\node[color=blue] at (14.3,1.3) {\Large{$A$}};
	\node[color=blue] at (17.5,1.3) {\Large{$C$}};

\end{tikzpicture} \caption{The first projection agent's valuation function. In interval $Out_1$ the labels are ordered $A,B,C$ by construction. In interval $Temp_1$ the labels may appear in any order. The example above illustrates the case where the ordering is $B,A,C$. Note that even if the second cut in $Out_1$ moves to the right or left (as long as $Out_1$ remains well-cut), the second cut in $Temp_1$ will not move. This is a crucial property of the projection agent.}
\label{fig:projection-gadget}
\end{figure}

In order to show that \cref{lem:ppad-projection} holds for $In_1$ and $Out_1$, it suffices to use a $G_{\ttimes -1}$-gate (see next section) with input $Temp_1$ and output $In_1$.

Now let $I := Out_2$ and $O := Temp_2$. The second projection agent's valuation function has height $3/10$ in $O[1,2] \cup O[4,5] \cup O[7,8]$, height $1/120$ in $X(O)$, height $1/30$ in $I[0,1/2]$, height $1/120$ in $I[5,7]$, height $1/60$ in $I[17/4,19/4] \cup I[17/2,9]$, and height $0$ everywhere else. In other words, the second projection agent's valuation function in interval $Out_2$ is the same as the first projection agent's in $Out_1$, except that it is mirrored. By using the same arguments we used for the first projection agent, we get that \cref{lem:ppad-projection} also holds for $In_2$ and $Out_2$.

\subsubsection{Arithmetic gates}

In this section, we prove \cref{lem:ppad-gates}, by providing an explicit construction of the arithmetic gates.
As before, let $T : \mathbb{R} \to [-1,1]$ denote the truncation operator. When constructing the agents implementing the gates, we can assume that their output interval is well-cut, since this holds for all agents by construction.

\paragraph{\textbf{Multiplication by a constant $\zeta \in \mathbb{Q}$ $[G_{\ttimes \zeta}]$}:} Let $I$ and $O$ be two disjoint intervals of length $9$. First consider the case $\zeta \leq 0$. We create an agent with the following valuation function. The agent's density function has height $3/10$ in $O[1,2] \cup O[4,5] \cup O[7,8]$, height $\frac{|\zeta|}{60(|\zeta|+1)}$ in $X(I)$, height $\frac{1}{60(|\zeta|+1)}$ in $X(O)$, and height $0$ everywhere else. If $I$ is a valid encoding, then $O$ is a valid encoding of the value $v(O) = T[v(I) \times \zeta]$. If $\zeta > 0$, then use a $G_{\ttimes -\zeta}$-gate and then a $G_{\ttimes -1}$-gate.

Let us see why this works. Since $I$ encodes a valid value, exactly $1/3$ of the agent's value in interval $I$ goes to label $C$. It follows that exactly $1/3$ of the agent's value in interval $O$ must go to label $C$. By construction of the valuation function in $O$, it follows that $1/3$ of $X(O)$ must go to label $C$. As a result, $O$ is also a valid encoding. Now, if label $C$ gets the left-most piece in $O$, then there is a cut at position $O[3]$. The other cut, which is located in $O[19/4,29/4]$, will thus encode the value of interval $O$ and the truncation will work similarly to our Consensus-Halving gadgets. The analogous argument also holds if $C$ gets the right-most piece in $O$. If label $C$ gets the middle piece in $O$, then the distance between the two cuts will be exactly $3$. Thus, the two cuts ``move together'' and will touch a big block at the same time. As a result, the truncation works here as well.

\paragraph{\textbf{Addition $[G_{\tplus}]$}:} Let $I_1$, $I_2$ and $O$ be three pairwise disjoint intervals of length $9$. The agent's density function has height $3/10$ in $O[1,2] \cup O[4,5] \cup O[7,8]$, height $1/180$ in $X(I_1) \cup X(I_2) \cup X(O)$, and height $0$ everywhere else. If $I_1$ and $I_2$ are valid encodings, then $O$ is a valid encoding of the value $v(O) = - T[v(I_1) + v(I_2)]$. To obtain $G_{\tplus}$, just apply a $G_{\ttimes -1}$-gate on the output. A similar reasoning to the case of $G_{\ttimes \zeta}$ proves the correctness of this gate too.

\paragraph{\textbf{Constant $\zeta \in [-1,1] \cap \mathbb{Q}$ $[G_{\zeta}]$:}} For this we use the left-most interval of length $9$ of the domain, namely $Out_1$. We know that interval $I:=Out_1$ is the output interval of some gate, so it is well-cut. Furthermore, we know that the labels appear in order $A$, $B$, $C$ from left to right (by convention). Thus, we know that interval $I[0,1/2]$ is labeled $A$, interval $I[17/4,19/4]$ is labeled $B$, and interval $I[17/2,9]$ is labeled $C$.

Introduce an agent with output interval $O$ of length $9$. The agent's density function has height $3/10$ in $O[1,2] \cup O[4,5] \cup O[7,8]$, height $1/120$ in $X(O)$, height $1/30$ in $I[17/2,9]$, height $\frac{1-\zeta/2}{30}$ in $I[0,1/2]$, height $\frac{1+\zeta/2}{30}$ in $I[17/4,19/4]$, and height $0$ everywhere else. From this construction, it follows that interval $O$ is a valid encoding of the value $v(O) = \zeta$.\\

\noindent As a result, \cref{lem:ppad-gates} holds and \cref{thm:ppad3hard} follows. 
\section{Future directions}
The main technical question is whether $\varepsilon$-\ch\ for single-block valuations is PPA-hard or polynomially solvable, or perhaps even complete for some other class. Another interesting direction is to extend the PPA-hardness result of \cref{thm:ppahardness} (or even for a larger number of blocks) to constant $\varepsilon$; such a result however would seemingly require radically new ideas, namely an averaging argument over a constant set of outputs that is robust to stray cuts. 

In a slightly different direction, \citet{deligkas2020company} very recently showed that the problem for a \emph{constant} number of agents is PPA-complete, if we allow agents to have more general valuations, in particular non-additive. This leaves open the fundamental question of showing PPA-hardness of $\varepsilon$-\ch\ for a constant number of agents with additive valuations. 

Finally, it would be interesting to study the complexity of the \cd\ problem when $k\geq 3$ and possibly strengthen or extend our hardness result to other values of $k$. Given the membership of the problem in PPA-$k$ \citep{FHSZ20}, for any $k$ which is a prime power, the important question is whether \cd\ (and consequently Necklace Splitting with $k$ thieves \citep{FRG18-Consensus}) is actually complete for PPA-$k$.

\subsubsection*{Acknowledgments}
Alexandros Hollender was supported by an EPSRC doctoral studentship (Reference 1892947). Katerina Sotiraki was supported in part by NSF/BSF grant \#1350619, an MIT-IBM grant, and a DARPA Young Faculty Award,  MIT Lincoln Laboratories and Analog Devices. Part of this work was done while the author was visiting the Simons Institute for the Theory of Computing. Manolis Zampetakis was supported by a Google Ph.D. Fellowship and NSF Award \#1617730.  \\

\noindent We would like to thank Paul Goldberg for helpful discussions and comments, as well as the anonymous reviewers at EC'20 and SICOMP for their suggestions that helped improve the presentation of the paper.

\bibliographystyle{plainnat}
\bibliography{Consensus_easier}

\end{document}